\documentclass[12pt]{article}
\usepackage{setspace} 
\usepackage[]{algorithm2e}
\usepackage{fancyhdr}
\usepackage{geometry}
\usepackage{pdfpages}

\usepackage{latexsym, amssymb, amscd, amsthm, amsxtra, amsmath,amsthm,mathtools }
\usepackage{graphics, graphicx, color}
\usepackage{natbib}
\usepackage{ifpdf} 
\usepackage{multirow}
\graphicspath{{./images/}}

\usepackage{graphicx}
\usepackage{caption}
\usepackage{subcaption}
\usepackage{authblk}

\newtheorem{thm}{Theorem}

\newtheorem{lem}[thm]{Lemma}
\newtheorem{prop}[thm]{Proposition}
\theoremstyle{definition}

\newtheorem*{defn*}{Definition}
\theoremstyle{remark}
\newtheorem{remark}{Remark}

\newcommand{\norm}[1]{\|#1\|}

\newcommand{\Real}{\Re}

\newcommand{\vvec}{\mathrm{vec}}

\newcommand{\RR}{\Re}

\newcommand{\bA}{\textbf{A}}
\newcommand{\bB}{\textbf{B}}
\newcommand{\bM}{\textbf{M}}
\newcommand{\bs}{\textbf{s}}
\newcommand{\by}{\textbf{y}}
\newcommand{\bLambda}{\boldsymbol{\Lambda}}
\newcommand{\bkappa}{\boldsymbol{\kappa}}

\newcommand{\bzeta}{\boldsymbol{\zeta}}
\newcommand{\bphi}{\boldsymbol{\phi}}
\newcommand{\bmu}{\boldsymbol{\mu}}

\def\argmin{\mathop{\rm argmin}}

\def\half{\frac{1}{2}}

\def\sv{\mathbf s}

\def\xv{\mathbf x}

\def\Bv{\mathbf B}

\newcommand{\kappav}{\mbox{\boldmath{$\kappa$}}}

\newcommand{\muv}{\mbox{\boldmath{$\mu$}}}
\newcommand{\nuv}{\mbox{\boldmath{$\nu$}}}

\newcommand{\Lambdav}{\mbox{\boldmath{$\Lambda$}}}

\newcommand{\Psiv}{\mbox{\boldmath{$\Psi$}}}

\newcommand{\Cc}{\mathcal{C}}

\newcommand{\Lc}{\mathcal{L}}

\newcommand\commentout[1]{}

\title{Small sphere distributions for directional data \\
           with application to medical imaging}

\author[1]{Byungwon Kim}
\author[2]{Stephan Huckemann}
\author[3]{J\"{o}rn Schulz}
\author[1]{Sungkyu Jung}
\affil[1]{Department of Statistics, University of Pittsburgh}
\affil[2]{Felix-Bernstein-Institute for Mathematical Statistics in the Biosciences, University of G\"{o}ttingen}
\affil[3]{Department of Electrical and Computer Engineering, University of Stavanger}

\date{\today}
\begin{document}
\maketitle

\begin{abstract}
We propose new small-sphere distributional families for modeling multivariate directional data on $(\mathbb{S}^{p-1})^K$ for $p \ge 3$ and $K \ge 1$. In a special case of univariate directions in $\Re^3$, the new densities model random directions on $\mathbb{S}^2$ with a tendency to vary along a small circle on the sphere, and with a unique mode on the small circle. The proposed multivariate densities enable us to model association among multivariate directions, and are useful in medical imaging, where multivariate directions are used to represent shape and shape changes of 3-dimensional objects. When the underlying objects are rotationally deformed under noise, for instance, twisted and/or bend, corresponding directions tend to follow the proposed small-sphere distributions. The proposed models have several advantages over other methods analyzing small-circle-concentrated data, including inference procedures on the association and small-circle fitting. We demonstrate the use of the proposed multivariate small-sphere distributions in analyses of skeletally-represented object shapes and human knee gait data.
\end{abstract}

\textbf{Keywords}:
Bingham-Mardia distribution, directional data,  skeletal representation, small circle, small sphere, likelihood ratio test, maximum likelihood estimation, von Mises-Fisher distribution.

\section{Introduction}

In medical imaging, accurately assessing and correctly diagnosing shape changes of internal organs is a major objective of a substantial challenge. Shape deformations can occur through long-term growth or necrosis as well as by short-term natural deformations. In view of surgery and radiation therapy, it is important to model all possible variations of object deformations by both long- and short-term changes, in order to control the object's exact status and shape at treatment time. \emph{Rotational deformations} such as rotation, bending, and twisting form a key sub-category of possible shape changes. For instance, shape changes of hippocampi in the human brain have been shown to mainly occur in the way of bending and twisting \citep{Joshi2002,Pizer2013}.

For the task of modeling 3D objects an abundance of approaches have been introduced. Closely related to our work are landmark-based shape models \citep{Cootes1992, Dryden1998, Kurtek2011} where a solid object is modeled by the positions of surface points, chosen either anatomically, mathematically or randomly.
A richer family of models is obtained by attaching directions  normal to the sampled surface points. More generally, in skeletal representations \citep[called \emph{s-reps}, ][]{Pizer2008}, an object is modeled by the combination of skeletal positions (lying on a medial sheet inside of the object) and spoke vectors (connecting the skeletal positions with the boundary of the object). In these models, describing the variation of rotational deformations can be transformed into a problem of exploring the motion of \emph{directional vectors} on the unit two-sphere. As argued in \cite{Schulz2015}, directional vectors representing rotational deformations tend to be concentrated on small circles on the unit sphere; a toy data example in Fig.~\ref{fig:smallcircle} shows a typical pattern of such observations.

\begin{figure}[t!]
	\begin{center}
	\begin{subfigure}{0.32\textwidth}
		\includegraphics[width = \textwidth]{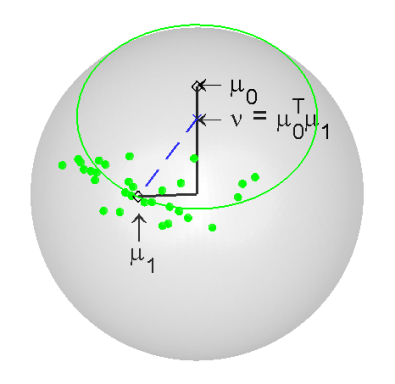}
		\caption{Data near a small circle}
		\label{fig:toyexample}
	\end{subfigure}
	\begin{subfigure}{0.32\textwidth}
		\includegraphics[width = \textwidth]{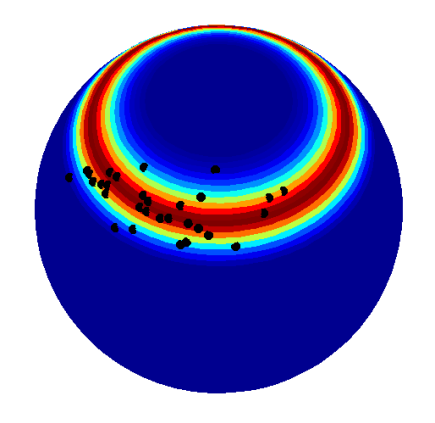}
		\caption{Bingham-Mardia}
		\label{fig:withBMdensity}
	\end{subfigure}
	\begin{subfigure}{0.32\textwidth}
		\includegraphics[width = \textwidth]{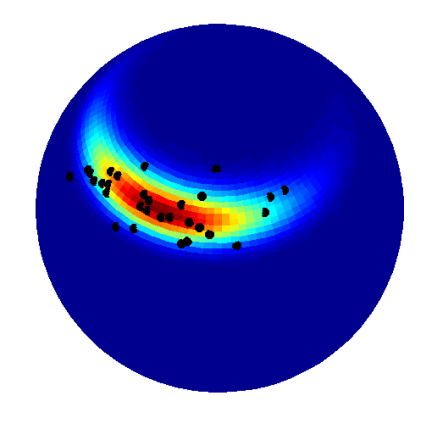}
		\caption{Proposed density}
		\label{fig:withS1density}
	\end{subfigure}
        \caption{(a) Toy example showing observations (solid green) distributed near a small circle $\mathcal{C}(\mu,\nu)$. The heat maps of fitted Bingham-Mardia density (b) and the proposed small-sphere density of the first kind (c) are overlaid. Red: high density, blue: low density. \label{fig:smallcircle}}
	\end{center}
\end{figure}

Motivated by the analysis of such s-rep data, we propose new distributional families and their multivariate extensions in order to model such directional data on the unit sphere $\mathbb{S}^{p-1} = \{x \in \Real^{p} \,|\, \|x\| = 1\}$ in arbitrary dimension $p \ge 3$. (Throughout the paper, $\|x\| = ( x^\top x)^{1/2}$ is the usual 2-norm of the vector $x$.) To precisely describe the targeted data situation, we define a $(p-2)$-dimensional \emph{subsphere} of $\mathbb{S}^{p-1}$ as the set of all points equidistant from $\mu \in \mathbb{S}^{p-1}$, denoted by
$$ \mathcal{C}(\mu,\nu) = \{ x \in \mathbb{S}^{p-1} \,|\, \delta(\mu,x) = \arccos(\nu) \}, \quad \nu \in (-1,1).$$
Here, $\delta(u,v)= \arccos(u^\top v)$ is the geodesic distance between $u, v \in \mathbb{S}^{p-1}$.
The subsphere is called a \emph{great subsphere} if $\nu = 0$ and a proper \emph{small subsphere} if $\nu \neq 0$. Note that $\mathcal{C}(\mu,\nu) \subset \mathbb{S}^{p-1}$ is well-defined for all $p > 1$. For the special case of $p = 3$, $\mathcal{C}(\mu,\nu)$ is a circle, a \emph{great circle} if $\nu = 0$, and a proper \emph{small circle} if $\nu \neq 0$.
To model the data in Fig.~\ref{fig:smallcircle}, one may naively use the Bingham-Mardia distribution \citep[called ``\emph{BM}'' hereafter, ][]{Bingham1978}, which is a family of densities on $\mathbb{S}^2$ with a modal ridge along a small circle. However, typical observations we encountered in applications do not uniformly spread over the full circle, and the BM distribution does not fit well, as shown in Fig.~\ref{fig:smallcircle}(b). 
Moreover, when by a single observation multiple directional vectors are provided, that is, data are on a polysphere $(\mathbb{S}^2)^K$, to the knowledge of the authors, there is no  tool available to date, to model  dependencies between directions.

In this paper, we propose two types of new distributional families for random directional vectors on $\mathbb{S}^{p-1}$, which we call \emph{small-sphere distributions} of the first ($\textrm{S1}$) and second ($\textrm{S2}$) kind. If $p = 3$, the proposed distributions may be called \emph{small-circle distributions}.
These two distributional families are designed to have higher densities on $\mathcal{C}(\mu,\nu)$ and to have a unique mode on $\mathcal{C}(\mu,\nu)$. An example of a small-sphere density, fitted to the toy data is shown in Fig.~\ref{fig:smallcircle}(c).
The new densities are natural extensions of the BM distribution with an additional term explaining a decay from a mode. If the additional term is a von Mises-Fisher (vMF) density on $\mathbb{S}^{p-1}$, we obtain the S1, which is a subfamily of the general Fisher-Bingham distribution \citep{Mardia1975, Mardia2000}. On the other hand, if the additional term is a vMF density on the small sphere ($\cong \mathbb{S}^{p-2}$), we obtain the S2 distribution, in which case the ``horizontal'' and ``vertical'' components of the directional vectors are independent of each other. 

Several multivariate extensions of the new distributions to $(\mathbb{S}^{p-1})^K$, $K \ge 2$, are discussed as well. In particular, we show that a special case, called MS2, of our multivariate extensions is capable of modeling \emph{dependent} random vectors. It has a straightforward interpretation, and we provide for fast estimation of its parameters. This MS2 distribution is specifically designed with s-rep applications in mind. In particular, s-rep data from rotationally-deformed objects have directional vectors that are ``rotated together,'' share a common axis of rotation, and are ``horizontally dependent'' (when the axis is considered to be vertically positioned). The component-wise independence of the S2 distributions plays a key role in this simple and interpretable extension. 
We discuss here likelihood-based parameter estimation and testing procedures of the multivariate distributions.

The contribution of this paper is summarized in Table~\ref{tab:summary}. While the new distributions contribute to the literature of directional distributions \citep{Mardia2000}, the proposed estimation procedures for the S1, S2 and MS2 parameters can be thought of as a method of fitting small-subspheres to data, which has been of separate interest. Nonparametric least-squares type solutions for such problem dates back to  \cite{Mardia1977}, \cite{GrayGeiser1980}, and \cite{Rivest1999}. \cite{Jung2012} proposed to recursively fitting small-subspheres in dimension reduction of directional and shape data. \cite{Pizer2013} proposed to combine separate small-circle fitting results in the analysis of s-rep data. In a similar spirit, \cite{Jung2011} and \cite{Schulz2015} also considered fitting  small-circles in applications to  s-rep analysis. In a simulation study, we show that our estimators provide smaller mean angular errors in small-circle fits than recent developments listed above. Moreover, our parametric framework provides a formal large-sample likelihood ratio test procedure, which is applicable to a number of hypothesis settings and is more powerful in testing dependence among multivariate directions than  competing methods.

\begin{table}
\centering
    \begin{tabular}{c|c|c|c}
                   & First kind (S1) & \multicolumn{2}{c}{ Second kind (S2)} \\
                   \hline
      Univariate   & S1              & \multicolumn{2}{c}{S2} \\
      Multivariate (indep.) & iMS1            & \multicolumn{2}{c}{iMS2}\\
      Multivariate (dep.) &     $\times$            & GMS2 & MS2 \\
      Simulation          & Gibbs sampling  &   $\times$   & Exact sampling \\
      Estimation          & Approximate m.l.e. & $\times$  & Approximate m.l.e. \\
      Hypothesis testing  & Likelihood ratio   & $\times$  & Likelihood ratio \\
     \end{tabular}
\caption{Small-sphere distributions (top three rows) and methods (bottom three rows) developed in this paper. Items with ``$\times$" mark are beyond the scope of this paper.}
\label{tab:summary}
\end{table}

The rest of the article is organized as follows. In Section~\ref{sec:small.circle.model}, we introduce the proposed densities of the S1 and S2 distributions and discuss their multivariate extensions including the MS2 distribution. Procedures of obtaining random variates from the proposed distributions are also proposed and discussed.
In Section~\ref{sec:estimation}, algorithms to obtain approximate maximum likelihood estimators of the parameters are  proposed and discussed. In Section~\ref{sec:testing}, we introduce several hypotheses of interest and procedures of large-sample approximate likelihood-ratio tests. In Section \ref{sec:simulation}, we empirically show that the proposed models are superior over other methods for small-circle estimation and for estimating dependencies in the multivariate setting, and show that the proposed testing procedures effectively prevent overfitting. 
In Sections~\ref{sec:application.examples-srep} and \ref{sec:applications.knee} we demonstrate applications of the new multivariate distributions to analyze models that represent human organs and knee motions. The appendix and the online supplementary material contain technical details and additional numerical results.

\section{Parametric small-sphere models}\label{sec:small.circle.model}

First we introduce two classical spherical densities, then we suitably combine them for our purposes.

\subsection{Two classical distributions on $\mathbb{S}^{p-1}$}\label{sec:backgrounds}
The von Mises-Fisher (vMF) distribution \citep[][p.168]{Mardia2000} is a fundamental unimodal and isotropic distribution for directions with density
\begin{equation}\label{eq:vmfdensity}
	f_{\textrm{vMF}}(x; \mu, \kappa) = \left( \frac{\kappa}{2} \right)^{p/2-1} \frac{1}{\Gamma(p/2)I_{p/2-1}(\kappa)} \exp\{ \kappa\mu^\top x \}, \quad x \in \mathbb{S}^{p-1}\,.
\end{equation}
Here, $\Gamma$ is the gamma function and $I_v$ is the modified Bessel function of the first kind and order $v$. The parameter $\mu \in \mathbb{S}^{p-1}$ locates the unique mode with $\kappa \geq 0$ representing the degree of concentration.

The Bingham-Mardia (BM) distribution was introduced by \cite{Bingham1978} to fit data in $\mathbb{S}^2$ that cluster near a small circle $\mathcal{C}(\mu,\nu)$. For an arbitrary dimension $p \ge 3$, the BM density is given by
\begin{equation}\label{eq:bmdensity}
	f_{\textrm{BM}}(x; \mu, \kappa, \nu) = \frac{1}{\alpha(\kappa,\nu)}\exp\{ -\kappa(\mu^\top x - \nu)^2\}, \quad x \in \mathbb{S}^{p-1},
\end{equation}
where $\alpha(\kappa,\nu) >0$ is the normalizing constant.

For our purpose of generalizing these distributions, we represent the variable $x \in \mathbb{S}^{p-1}$, $p\geq 3$,  by spherical angles $\phi_1,\ldots,\phi_{p-1}$ satisfying $\cos\phi_1= \mu^\top x$. Setting $s := \cos\phi_1 \in [-1,1]$ and  $\phi:=(\phi_2,\ldots,\phi_{p-1}) \in [0,\pi]^{p-3}\times [0, 2\pi)$, the random vector $(s,\phi)$ following the von Mises-Fisher (\ref{eq:vmfdensity}) or Bingham-Mardia (\ref{eq:bmdensity}) distribution has the respective density:
\begin{eqnarray}
	g_{\textrm{vMF}}(s,\phi; \kappa) &=& \left( \frac{\kappa}{2} \right)^{p/2-1} \frac{1}{\Gamma(p/2)I_{p/2-1}(\kappa)} \exp\{ \kappa s\}\,, \label{eq:vmf.canonical} \\
	g_{\textrm{BM}}(s,\phi; \kappa, \nu) &=& \frac{1}{\alpha(\kappa,\nu)}\exp\{ -\kappa(s - \nu)^2\}\,. \label{eq:bm.canonical}
\end{eqnarray}
In consequence, for both distributions,   $s$ and $\phi$ are independent, and the marginal distribution of $\phi$, which parametrizes a co-dimension 1 unit sphere $\mathbb{S}^{p-2}$, is uniform. In (\ref{eq:vmf.canonical}), the marginal distribution of $s$ is a shifted exponential distribution truncated  to $s \in [-1,1]$, while in (\ref{eq:bm.canonical}) the marginal distribution of $s$ is a normal distribution truncated to $s \in [-1,1]$. Both densities are isotropic, i.e. rotationally symmetric with respect to $\mu$. The vMF density is maximal at the mode $\mu$ and decreases as the latitude $\phi_1$ increases, while the BM density is uniformly maximal on the small-sphere $\mathcal{C}(\mu,\nu)$ and decreases as $\phi_1$ deviates from $\arccos (\nu)$.


\subsection{Small-sphere distributions of the first and second kind}\label{sec:define.S1.S2}

The proposed small-sphere densities of the first and second kind on $\mathbb{S}^{p-1}$, for $x=(x_1,\ldots,x_p) \in \mathbb{S}^{p-1}$ with parameters $\mu_0,\mu_1 \in \mathbb{S}^{p-1}$, $\nu = \mu_0^\top \mu_1 \in (-1,1)$, $\kappa_0 > 0$, $\kappa_1 > 0$, are given by
\begin{eqnarray}
    f_{\textrm{S1}}(x; \mu_0, \mu_1, \kappa_0, \kappa_1) &=& \frac{1}{a(\kappa_0, \kappa_1, \nu)} \exp\{-\kappa_0(\mu_0^\top x - \nu)^2 + \kappa_1 \mu_1^\top x\}\,, \label{eq:S1density} \\
    f_{\textrm{S2}}(x; \mu_0, \mu_1, \kappa_0, \kappa_1) &=& \frac{1}{b(\kappa_0, \kappa_1, \nu)} \exp\left\{ -\kappa_0(\mu_0^\top x - \nu)^2 + \kappa_1 \frac{\mu_1^\top P_{\mu_0} x}{\sqrt{\mu_1^\top P_{\mu_0} \mu_1 x^\top P_{\mu_0} x} } \right\} \,, \label{eq:S2density}
\end{eqnarray}
respectively, where $a(\kappa_0, \kappa_1, \nu)$ and $b(\kappa_0, \kappa_1, \nu)$ are normalizing constants. Here, $P_{\mu_0}$ denotes the matrix of orthogonal projection to the orthogonal \emph{complement} of $\mu_0$; $P_{\mu_0} = I_p - \mu_0\mu_0^\top$, where $I_p$ is the identity matrix. (In (\ref{eq:S2density}), we use the convention $0/0 = 0$.)

These distributions are well-suited to model observations that are concentrated near the small sphere $\mathcal{C}(\mu_0,\nu)$ but are not rotationally symmetric.
The first kind (\ref{eq:S1density}) is a natural combination of the vMF (\ref{eq:vmfdensity}) and BM (\ref{eq:bmdensity}) distributions.
The parameter $\mu_0$ represents the axis of the small sphere $\mathcal{C}(\mu_0,\nu)$, while $\mu_1$ gives the mode of the distribution, which, by the definition of $\nu$, is on the small sphere $\mathcal{C}(\mu_0,\nu)$. These parameters, $\mu_0, \mu_1$, $\nu$, are illustrated in Fig.~\ref{fig:smallcircle}(a) for the $p = 3$ case. The parameter $\kappa_0$ controls the \emph{vertical concentration} towards the small sphere (with an understanding that $\mu_0$ is arranged vertically).
In (\ref{eq:S1density}),  $\kappa_1$ controls the isotropic part of the concentration around the mode, forcing the density to decay from $\mu_1$. 

The rationale for the second kind (\ref{eq:S2density}) is better understood with a change of variables. Let us assume for now that $\mu_0 = (1,0,\ldots,0)^\top$. For any $x = (x_1,\ldots,x_p)^\top \in \mathbb{S}^{p-1}$, write $s := x_1 = \mu_0^\top x$.
If the spherical coordinate system $(\phi_1,\ldots,\phi_{p-1})$ as defined for (\ref{eq:bm.canonical}) is used, then $s = \cos \phi_1$. The ``orthogonal complement'' of $s$ is denoted by
\begin{equation}\label{eq:y_in_spherical}
y := (x_2,\ldots,x_p) / \sqrt{1-s^2} \in \mathbb{S}^{p-2},
\end{equation}
where the vector $y$ is obtained from the relation $P_{\mu_0}x/\|P_{\mu_0}x\| = (0,y) \in \mathbb{S}^{p-1}$. Similarly, define $\widetilde{\mu_1}\in \mathbb{S}^{p-2}$ as the last $p-1$ coordinates of $P_{\mu_0}\mu_1/\|P_{\mu_0}\mu_1\|$.
Then the random vector $(s,y)\in [-1,1]\times \mathbb{S}^{p-2}$ from the S1 or S2 has densities
\begin{eqnarray}\label{eq:s1canonical}
	g_{\textrm{S1}}(s,y; \mu_1,\kappa_0, \kappa_1) &=& \frac{1}{a(\kappa_0, \kappa_1, \nu)} \exp\left\{-\kappa_0(s - \nu)^2 + \kappa_1 \mu_1^\top \left(s,\sqrt{1-s^2}y\right)\right \}\,,\\
\label{eq:s2canonical}
	g_{\textrm{S2}}(s,y; \mu_1,\kappa_0, \kappa_1) &=& \frac{1}{b(\kappa_0, \kappa_1, \nu)} \exp\left\{ -\kappa_0(s - \nu)^2 + \kappa_1 \widetilde{\mu_1}^\top y \right\}\,,
\end{eqnarray}
respectively, for $s \in [-1,1],\, y\in \mathbb{S}^{p-2}$.
The subtle difference is that for the first kind (\ref{eq:s1canonical}), the ``vMF part'' (the second term in the exponent) is not statistically independent from the ``BM part'', while  it is true for the second kind (\ref{eq:s2canonical}). That is, $s$ and $y$ are independent only in the second kind. Accordingly, in  (\ref{eq:s2canonical}), $\kappa_1$ controls the \emph{horizontal concentration} towards the mode $\mu_1$.
The parameters $\mu_0, \mu_1$ and $\kappa_0$ of the second kind have the same interpretations as those of the first kind.

We use the notation $X \sim \textrm{S1}(\mu_0,\mu_1,\kappa_0,\kappa_1)$ and $Y\sim \textrm{S2}(\mu_0,\mu_1,\kappa_0,\kappa_1)$ for random directions $X,Y \in \mathbb{S}^{p-1}$ following small-sphere distributions of the first and second kind with parameters $(\mu_0,\mu_1,\kappa_0,\kappa_1)$, respectively.
The proposed distributions  are quite flexible and can fit a wide range of data. In Figure~\ref{fig:S1.examples}, we illustrate the S1 densities (\ref{eq:S1density}) with various values of the concentration parameters $\kappa_0, \kappa_1$. In all cases, the density is relatively high near the small circle $\mathcal{C}(\mu_0,\nu)$ and has the mode at $\mu_1 \in \mathcal{C}(\mu,\nu)$. Despite the difference in their formulations, the S2 densities (\ref{eq:S2density}) look similar to S1 densities for each fixed parameter-set. We refer to the online supplementary material for several visual examples of the S2 density.

\begin{figure}[tb!]
	\begin{center}
	\begin{subfigure}{0.2\textwidth}
		\includegraphics[width = \textwidth]{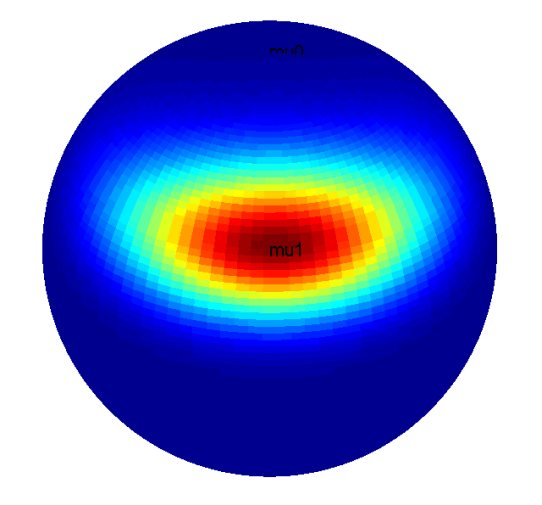}
		\caption{$\kappa_0 = 10$, $\kappa_1 = 4$}
	\end{subfigure}
	\begin{subfigure}{0.2\textwidth}
		\includegraphics[width = \textwidth]{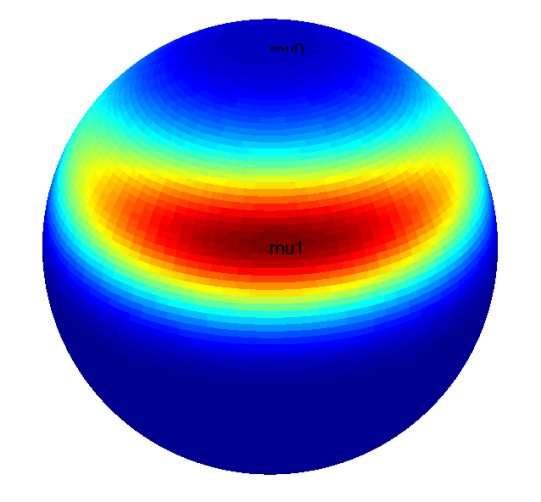}
		\caption{$\kappa_0 = 10$, $\kappa_1 = 1$}
	\end{subfigure}
	\begin{subfigure}{0.2\textwidth}
		\includegraphics[width = \textwidth]{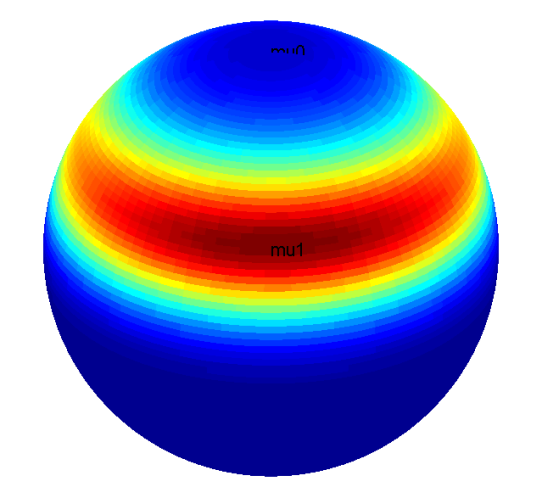}
		\caption{$\kappa_0 = 10$, $\kappa_1 = .5$}
	\end{subfigure} \\
	\begin{subfigure}{0.2\textwidth}
		\includegraphics[width = \textwidth]{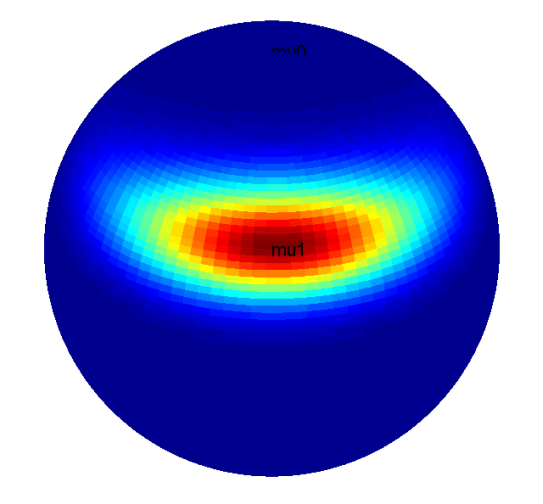}
		\caption{$\kappa_0 = 20$, $\kappa_1 = 4$}
	\end{subfigure}
	\begin{subfigure}{0.2\textwidth}
		\includegraphics[width = \textwidth]{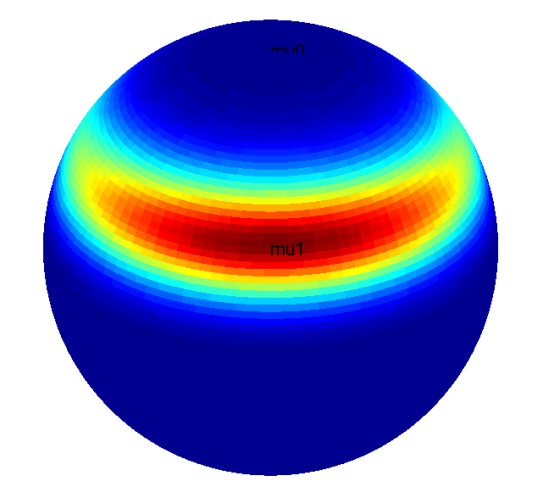}
		\caption{$\kappa_0 = 20$, $\kappa_1 = 1$}
	\end{subfigure}
	\begin{subfigure}{0.2\textwidth}
		\includegraphics[width = \textwidth]{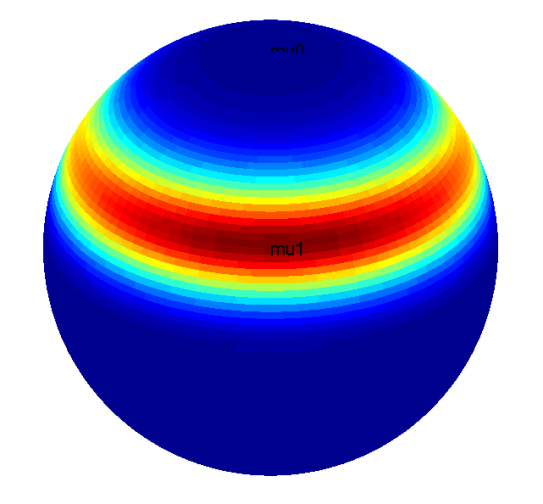}
		\caption{$\kappa_0 = 20$, $\kappa_1 = .5$}
	\end{subfigure} \\
	\begin{subfigure}{0.2\textwidth}
		\includegraphics[width = \textwidth]{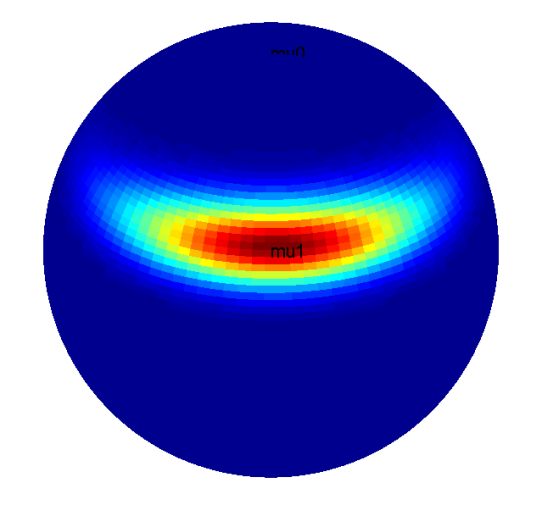}
		\caption{$\kappa_0 = 40$, $\kappa_1 = 4$}
	\end{subfigure}
	\begin{subfigure}{0.2\textwidth}
		\includegraphics[width = \textwidth]{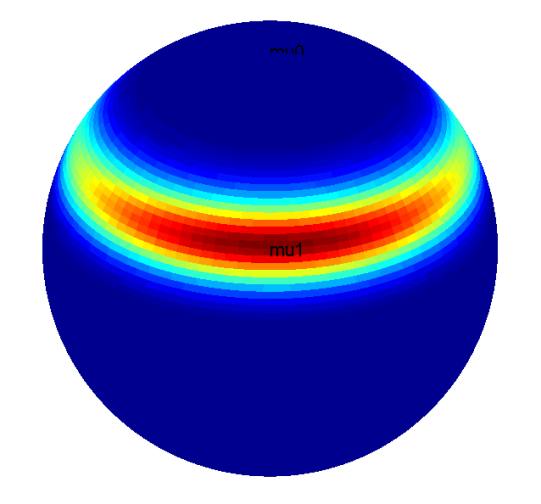}
		\caption{$\kappa_0 = 40$, $\kappa_1 = 1$}
	\end{subfigure}
	\begin{subfigure}{0.2\textwidth}
		\includegraphics[width = \textwidth]{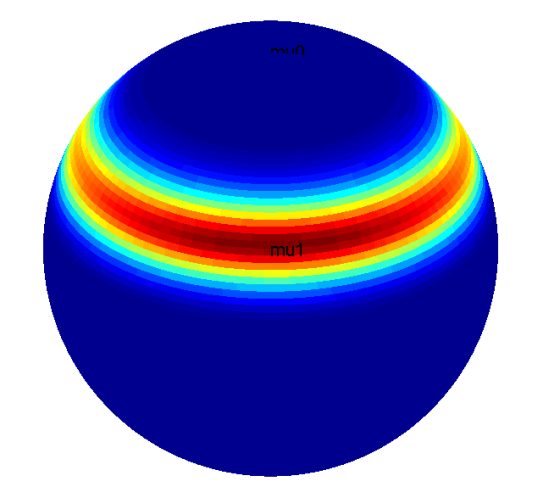}
		\caption{$\kappa_0 = 40$, $\kappa_1 = .5$}
	\end{subfigure}
		\caption{The S1 densities on $\mathbb{S}^{2}$ modeling non-isotropic small-circle distributions. High density (red), low density (blue). In all figures, $\mu_0$ points to the north pole and $\mu_1$ satisfies $\mu_0^\top \mu_1 = 1/2$.
Rows and columns correspond to different choices of concentration parameters ($\kappa_0$,  $\kappa_1$).\label{fig:S1.examples}}
	\end{center}
\end{figure}

The following invariance properties are proven in the Appendix.
\begin{prop}\label{prop:invariance}
    Let $X,Y \in \mathbb{S}^{p-1}$ be random directions with $X \sim \textrm{S1}(\mu_0,\mu_1,\kappa_0,\kappa_1)$ and $Y \sim \textrm{S2}(\mu_0,\mu_1,\kappa_0,\kappa_1)$ and let $B$ be a full rank $p \times p$ matrix.
    \begin{itemize}
        \item[(i)]
        If $B\mu_0 = \mu_0$ and $B\mu_1 = \mu_1$ then $X$ and $BX$ have the same distribution and so do $Y$ and $BY$.
        \item[(ii)]
        If $B$ is orthogonal and  $X,BX$ have the same distribution (or $Y,BY$ have the same distribution), then $B\mu_0 = \mu_0$
        and $B\mu_1 = \mu_1$.
        \item[(iii)]
        $X \sim \textrm{S1}(-\mu_0,\mu_1,\kappa_0,\kappa_1)$ and $Y \sim \textrm{S2}(-\mu_0,\mu_1,\kappa_0,\kappa_1)$.
    \end{itemize}
\end{prop}

An example for the matrix $B$ in Proposition~\ref{prop:invariance}(i) is the {reflection} matrix $B = I_p - 2UU^\top$, where $U = [u_3, \ldots, u_p]$ is such that $[u_1,\ldots,u_p]$ is a $p\times p$ orthogonal matrix whose column vectors $u_1$ and $u_2$ generate $\mu_0$ and $\mu_1$.

\begin{remark}\label{remark:S1toFB}
The S1 distribution is a special case of the Fisher-Bingham distribution \citep{Mardia1975}. Following the notation of \cite{kent1982fisher}, the S1 distribution may be labeled as a \emph{FB$_6$ distribution}, in the special case of $p=3$, emphasizing the 6-dimensional parameter space.
In terms of the general parameterization of the Fisher-Bingham density \citep[cf.][p.174]{Mardia2000}, we write $\gamma = 2\kappa_0\nu\mu_0 + \kappa_1\mu_1$ and $A = \kappa_0\mu_0\mu_0^\top$, so that the S1 density (\ref{eq:S1density}) is expressed as
\begin{equation}\label{eq:fbdensity}
    f_{\rm S1}(x; \gamma, A) = \frac{1}{\alpha(\gamma, A)} \exp\{\gamma^\top x - x^\top Ax\},
\end{equation}
where
\begin{equation}\label{eq:fb-constant}
   \alpha(\gamma, A) = a(\kappa_0,\kappa_1,\nu)\exp\left\{\kappa_0 \nu^2\right\}.
\end{equation}
This relation to the general Fisher-Bingham distribution facilitates random data generation and maximum likelihood estimation, shown later in Sections~\ref{sec:random.generation} and \ref{sec:estimation.S1}.
\end{remark}

\subsection{Multivariate extensions}\label{sec:mult.extension}

The univariate small-sphere distributions (\ref{eq:S1density}--\ref{eq:S2density})  are now extended to model a tuple of \emph{associated} random directions, $\textbf{X}  = (X_{(1)},\ldots,X_{(K)}) \in (\mathbb{S}^{p-1})^K$.  We confine ourselves to a special case where the marginal distributions of $X_{(k)}$ have a common ``axis'' parameter $\mu_0$, but relaxing this condition is straightforward. We begin by introducing multivariate small-sphere distributions for independent random directions, denoted by iMS1 and iMS2.

\paragraph{Independent extensions.}
Suppose that, in the $K$-tuple of random directions $\textbf{X}$, each $X_{(k)} \in \mathbb{S}^{p-1}$ is marginally distributed as S1$(\mu_0, \mu_{1(k)}, \kappa_{0(k)},\kappa_{1(k)})$. Throughout, we assume that $\nu_{(k)} = \mu_0^\top\mu_{1(k)} \in (-1,1)$ so that the underlying small spheres do not degenerate. If the components of $\textbf{X}$ are mutually independent, then the joint density evaluated at $\textbf{x} \in (\mathbb{S}^{p-1})^K$ is 
\begin{equation}\label{eq:indep.mv.S1}
   f_{\rm iMS1}(  \textbf{x} ) \propto \exp\left\{ \Gamma^\top\textbf{x} - \textbf{x}^\top \textbf{A} \textbf{x} \right\}.
\end{equation}
Here, $\Gamma = [\gamma_{(1)},\ldots,\gamma_{(K)}]^\top$, where $\gamma_{(k)} = 2\kappa_{0(k)}\nu_{(k)}\mu_0 + \kappa_{1(k)}\mu_{1(k)}$, and $\textbf{A} = D_{\kappa_0} \otimes (\mu_0\mu_0^\top)$, where $D_{\kappa_0} = \textrm{diag}(\kappa_{0(1)},\ldots,\kappa_{0(K)})$. Each marginal density is of the form (\ref{eq:fbdensity}).

If each component is marginally distributed as S2$(\mu_0, \mu_{1(k)}, \kappa_{0(k)},\kappa_{1(k)})$, then writing the density in terms of $(s,y)$ as done for (\ref{eq:s2canonical}) facilitates our discussion. First, we decompose each $x_{(k)}$ into $s_{(k)} = \mu_0^\top x_{(k)}\in [-1,1]$ and $y_{(k)}\in \mathbb{S}^{p-2}$ as defined in (\ref{eq:y_in_spherical}). Further, we denote by $\widetilde{\mu_{1(k)}}$  the scaled projection of $\mu_{1(k)}$ as done for the univariate case.
Then an independent multivariate extension for the $\textrm{S2}$ model can be expressed as the joint density of $\textbf{s} = (s_{(1)},\ldots,s_{(K)})$ and $\textbf{y} = (y_{(1)},\ldots,y_{(K)})$, 
\begin{equation}\label{eq:indep.mv.S2}
  g_{\rm iMS2}(\textbf{s} , \textbf{y})  \propto \exp\left\{ H^\top\textbf{s} - \textbf{s}^\top D_{\kappa_0} \textbf{s} + \textbf{M}^\top \vvec(\textbf{y})\right\},
\end{equation}
where $H = (2\kappa_{0(1)}\nu_{(1)},\ldots, 2\kappa_{0(K)}\nu_{(K)})$ and $\textbf{M} = \vvec(\kappa_{1(1)}\widetilde{\mu_{1(1)}}, \ldots, \kappa_{1(K)}\widetilde{\mu_{1(K)}})$ while $\vvec(\cdot)$ denotes the column-wise vectorization of a matrix.

\paragraph{Vertical and horizontal dependence.}

Based on (\ref{eq:indep.mv.S1}) and (\ref{eq:indep.mv.S2}), we now contemplate on dependent models.
Obviously, if we allow in (\ref{eq:indep.mv.S1}) nonzero offdiagonal entries of $\bA$, then we obtain a dependent modification of the $\textrm{S1}$ model. With our applications in mind, however, we aim at modeling a specific structure of dependence  that is natural to the variables $(\bs,\by)$ in (\ref{eq:indep.mv.S2}).\

If $s_{(1)},\ldots,s_{(K)}$ are dependent, we speak of \emph{vertical dependence}; if  $y_{(1)},\ldots,y_{(K)}$ are dependent, we speak of \emph{horizontal dependence}. In practice, when we deal with small-circle concentrated directional data, association among these vectors usually occurs along small-circles with independent vertical errors.
For example, when a 3D object is modeled by skeletal representations, as described in more detail in Section~\ref{sec:application.examples-srep} and visualized in Fig.~\ref{fig:ellipsoids-all}, a deformation of the object is measured by the movements of directional vectors on $\mathbb{S}^2$. When a single rotational deformation  (such as bending, twisting or rotation) occurs, all the directions move along small-circles with a common axis $\mu_0$. In this situation, the longitudinal variations along the circles are dependent of each other because nearby spoke vectors are under the effect of similar deformations. (Examples of such longitudinal dependencies can be found in Section~\ref{sec:application.examples-srep} as well as in \cite{Schulz2015}.)
Adding such a horizontal (or longitudinal) dependence to a multivariate $\textrm{S1}$ model requires a careful introduction and parametrization of the offdiagonal entries of $\bA$ in (\ref{eq:indep.mv.S1}). This is not straightforward, and we leave it for future work.
On the other hand, it is feasible to extend the $\textrm{S2}$ model  by generalizing the ``vMF part'' of $\textbf{y}$, the last term in the exponent of (\ref{eq:indep.mv.S2}), to a Fisher-Bingham type.

To this end, we introduce a parameter matrix $\bB$ to model general quadratics in $\vvec(\textbf{y})$. This allows to write the densities for a general multivariate small-sphere distribution of the second kind ($\textrm{GMS2}$) as follows:
\begin{eqnarray}\label{eq:GMS2}\nonumber
\lefteqn{ g_{\textrm{GMS2}}(\bs,\by;H,D_{\kappa_0},\bM,\bB)}\\
&=& \frac{1}{T_1(H,D_{\kappa_0})T_2(\bM,\bB)}\exp\left\{ H^\top\textbf{s} - \textbf{s}^\top D_{\kappa_0} \textbf{s} + \textbf{M}^\top \vvec(\textbf{y}) + \vvec(\by)^\top  \bB\,\vvec(\by)\right\}
\end{eqnarray}
where $H,D_{\kappa_0}$ and $\bM$ as defined in (\ref{eq:indep.mv.S2}), and $T_1(H,D_{\kappa_0})$ and $T_2(\bM,\bB)$ are normalizing constants.
%
We set
$\bB = (B_{k,l})_{k,l=1}^K$, $B_{k,l} = (b^{(k,l)}_{i,j})_{i,j=1}^{p-1},$
as a block matrix with vanishing blocks $B_{k,k} =0$ on the diagonal. The submatrix $B_{k,l}$ models the horizontal association between $y_{(k)}$ and  $y_{(l)}$. 
The fact that $z^\top \bB z = z^\top \bB^\top z=\frac{1}{2} z^\top (\bB+\bB^\top )z$ for any vector $z \in \RR^{(p-1)K}$ allows us to assume without loss of generality  that $\bB$ is symmetric.

\paragraph{An $\textrm{MS2}$ distribution on $(\mathbb{S}^2)^K$.} As a viable submodel for the practically important case $p = 3$, we propose to use a special form for the offdiagonal blocks $B_{k,l}$ of $\bB$. In particular, with $\lambda_{k,l}$ representing the degrees of association between $y_{(k)}$ and  $y_{(l)}$, we set
\begin{eqnarray}\label{eq:Mardia-MS2}
B^{k,l} &=&
    2\begin{pmatrix}
      \widetilde{\mu_{1(k)}} & \widetilde{\mu_{2(k)}} \\
    \end{pmatrix}
      \begin{pmatrix}
      0 & 0 \\
      0 & \lambda_{k,l}\\
    \end{pmatrix}
    \begin{pmatrix}
      \widetilde{\mu_{1(l)}} & \widetilde{\mu_{2(l)}} \\
    \end{pmatrix}^\top \\
     &=& 2\lambda_{k,l} ~\widetilde{\mu_{2(k)}}\widetilde{\mu_{2(l)}}^\top,  \nonumber
\end{eqnarray}
where $\begin{pmatrix}
      \widetilde{\mu_{1(k)}} & \widetilde{\mu_{2(k)}} \\
    \end{pmatrix}$ is the rotation matrix given by setting
$$\widetilde{\mu_{2(k)}} = \begin{pmatrix}
      0 & -1 \\
      1 & 0 \\
    \end{pmatrix} \widetilde{\mu_{1(k)}}.$$
The density (\ref{eq:GMS2}) with the above parameterization of $\bB$ will be  referred to as a multivariate S2 distribution (MS2) for data on $(\mathbb{S}^2)^K$; its angular representation will be derived in (\ref{eq:dep.str.mult.S2density}) below.

Our choice of the simple parametrization (\ref{eq:Mardia-MS2}) does not restrict the modeling capability of the general model (\ref{eq:GMS2}), and has some advantages in parameter interpretations and also in estimation.  To see this, we resort to use an angular representation for $\by$ (available to this $p=3$ case). For each $k$, define $\phi_{(k)}$ and $\zeta_{(k)}$  such that $y_{(k)} = (\cos \phi_{(k)},\sin \phi_{(k)})^\top$ and $\widetilde{\mu_{1(k)}} = (\cos\zeta_{(k)},\sin\zeta_{(k)})^\top$. Accordingly, the inner products appearing in (\ref{eq:GMS2}) can be expressed as
\begin{eqnarray*}
 \widetilde{\mu_{1(k)}}^\top y_{(k)} =  \cos(\phi_{(k)} - \zeta_{(k)}), \quad  \widetilde{\mu_{2(k)}}^\top y_{(k)} =  \sin(\phi_{(k)} - \zeta_{(k)}).
\end{eqnarray*}
     Let
     $\boldsymbol{\phi} = (\phi_{(1)},\ldots,\phi_{(K)})^\top$,
     $\boldsymbol{\zeta} = (\zeta_{(1)},\ldots,\zeta_{(K)})^\top$,
     $\boldsymbol{\kappa}_1 = (\kappa_{1(1)},\ldots,\kappa_{1(K)})^\top$,
\begin{eqnarray*}
	c(\boldsymbol{\phi},\boldsymbol{\zeta}) &=& \left( \cos(\phi_{(1)} - \zeta_{(1)}),\ldots,\cos(\phi_{(K)} - \zeta_{(K)}) \right)^\top, \\
	s(\boldsymbol{\phi},\boldsymbol{\zeta}) &=& \left( \sin(\phi_{(1)} - \zeta_{(1)}),\ldots,\sin(\phi_{(K)} - \zeta_{(K)}) \right)^\top\,,
\end{eqnarray*}
and   $\bLambda = (\lambda_{k,l})_{k,l=1}^K$ where $\lambda_{k,l} (= \lambda_{l,k})$ for $k \neq l$ is the association parameter used in (\ref{eq:Mardia-MS2}), and $\lambda_{k,k}$ is set to zero. The density of the MS2 distribution, in terms of $(\bs, \boldsymbol{\phi})$, is then
 \begin{eqnarray}\label{eq:dep.str.mult.S2density}\nonumber
\lefteqn{ g_{\textrm{MS2}}(\bs,\bphi; H,D_{\kappa_0},\bkappa_1,\bzeta, \bLambda)}\\ &=& \frac{1}{T_1(H,D_{\kappa_0})T_3(\bkappa_1, \bLambda)}\exp\left\{ H^\top\textbf{s} - \textbf{s}^\top D_{\kappa_0} \textbf{s} + \boldsymbol{\kappa}_1^\top c(\boldsymbol{\phi},\boldsymbol{\zeta}) + \frac{1}{2}s(\boldsymbol{\phi},\boldsymbol{\zeta})^\top\boldsymbol{\Lambda}s(\boldsymbol{\phi},\boldsymbol{\zeta}) \right).
\end{eqnarray}

From (\ref{eq:dep.str.mult.S2density}), it can be easily seen that the ``horizontal angles'' $\bphi$ follow the \emph{multivariate von Mises} distribution \citep{Mardia2008} and are independent of the vertical component $\bs$.
As we will see later in Section~\ref{sec:estimation.S2}, this facilitates estimation for the MS2 distributions. Moreover,
since
\begin{eqnarray}\lefteqn{
\boldsymbol{\kappa}_1^\top c(\boldsymbol{\phi},\boldsymbol{\zeta}) + \frac{1}{2}s(\boldsymbol{\phi},\boldsymbol{\zeta})^\top\boldsymbol{\Lambda}s(\boldsymbol{\phi},\boldsymbol{\zeta})
} \nonumber\\
&=&
\sum_{k=1}^K \kappa_{1(k)}\left(1- \frac{(\phi_k - \zeta_k)^2}{2} \right) + \frac{1}{2}\sum_{k=1}^K\sum_{k\neq l =1}^K \Big(\lambda_{kl} (\phi_k - \zeta_k)(\phi_l - \zeta_l) \Big) + o(\norm{\boldsymbol{\phi}-\boldsymbol{\zeta}}^2),
\label{eq:largeconcentration}
\end{eqnarray}
for large enough concentrations, $\bphi$ is approximately multivariate normal with mean $\bzeta$ and precision matrix $\Sigma^{-1}$, where  $(\Sigma^{-1})_{kk} = \kappa_{1(k)}$ and $(\Sigma^{-1})_{kl} = -\lambda_{kl}$ for $1\leq k\neq l\leq K$. These parameters are naturally interpreted as partial variances and correlations. This interpretation of the parameters as entries of a precision matrix is most immediate under the  MS2, but is not under  the general case.

\subsection{Random data generation}\label{sec:random.generation}

Generating pseudo-random samples from the $\textrm{S1}$  and $\textrm{S2}$ distributions are important in simulations and in developments of computer-intensive inference procedures.

For simulation of the S1 (\ref{eq:S1density}) and iMS1 (\ref{eq:indep.mv.S1}) distribution, the fact that each marginal distribution of the iMS1 is a special case of the Fisher-Bingham is handy. Thereby, one can use the Gibbs sampling procedure developed for generating Fisher-Bingham-variate samples \cite{Hoff2009}.

For simulation of the S2 (\ref{eq:S2density}), iMS2 (\ref{eq:indep.mv.S2}), and MS2 (\ref{eq:dep.str.mult.S2density}) distribution, we take advantage of the independence between the pair ($\bs,\by$). As we assume vertical independence (i.e., $s_{(1)},\ldots,s_{(K)}$ are independent), each $s_{(k)}$ can be sampled separately. Therefore, sampling from the MS2 distribution amounts to independently drawing samples from a truncated normal distribution (for $s_{(k)}$) and from a multivariate von Mises distribution (for $\by$). Specifically, to sample $\xv = (x_{(1)},\ldots, x_{K)})$ from $\textrm{MS}2(\mu_0, \bmu_1, \bkappa_0,\bkappa_1,\bLambda)$, the following procedure can be used.

Step 1. For each $k$, sample $s_{(k)}$ from the truncated normal distribution with mean $\nu_{(k)}$ and variance $1/(2\kappa_{0(k)})$, truncated to  the interval (-1,1).

Step 2. For the S2 or iMS2 model, sample each $y_{(k)} \in \mathbb{S}^{p-2}$ in $\by = (y_{(1)},\ldots,y_{(K)})$ independently from the von Mises distribution with mean $(1,0,\ldots,0)$ and concentration $\kappa_{1(k)}$; for the MS2 distribution (when $p=3$), sample the $K$-tuple $\by \in (\mathbb{S}^{1})^K $ directly from the multivariate von-Mises distribution with mean $(1,0)$ and precision parameters  $\bkappa_1$ and $\bLambda$.

Step 3. For each $k$, let $E_{(k)}$ be a $p\times p$ orthogonal matrix with $(\mu_0, P_{\mu_0} \mu_{1(k)} / \| P_{\mu_0} \mu_{1(k)} \|)$ being the first two column vectors. Set $x_{(k)} = E_{(k)}^\top \left( s_{(k)}, (1 - s_{(k)}^2)^{-1/2} y_{(k)}  \right)$.

In our experiments, sampling from the S2 and MS2 distributions is much faster than from the S1. In particular, when the dimension $p$ or the concentration level is high, the Markov chain simulations for the S1 appear to be sluggish.
 Some examples of random samples from the S2, iMS2 and MS2 distributions are shown in Fig.~\ref{fig:random.S2}. The small-circles $C(\mu_0, \nu_{(k)})$ are also overlaid in the figure. Notably,  the MS2 sample in the rightmost panel clearly show a horizontal dependence. 
%
\begin{figure}[t!]
	\begin{center}
	\begin{subfigure}{0.3\textwidth}
		\includegraphics[width = \textwidth]{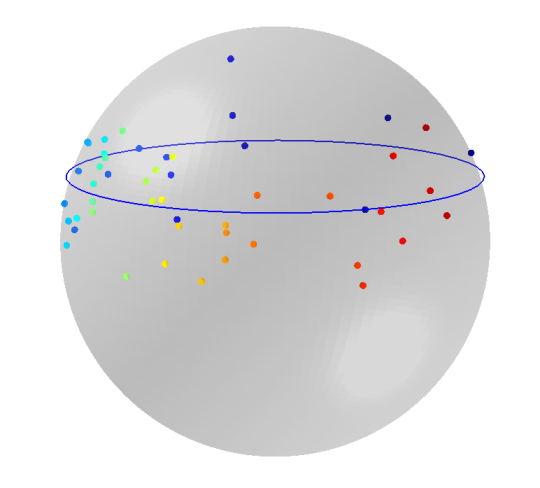}
		\caption{S2 on $\mathbb{S}^2$}
		\label{fig:S2.low.concentration}
	\end{subfigure}
	\begin{subfigure}{0.3\textwidth}
		\includegraphics[width = \textwidth]{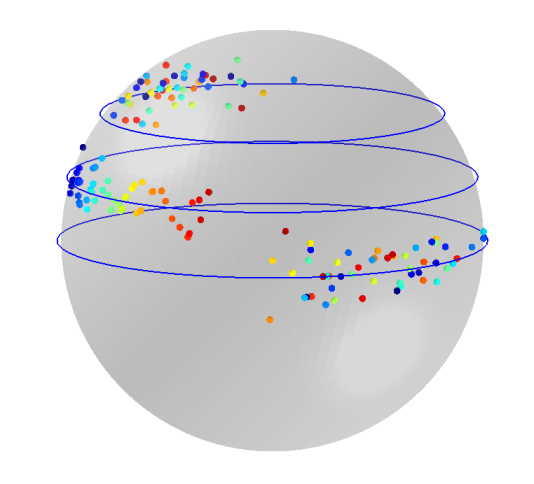}
		\caption{iMS2 on $(\mathbb{S}^2)^3$}
		\label{fig:S2.high.concentration}
	\end{subfigure}
	\begin{subfigure}{0.3\textwidth}
		\includegraphics[width = \textwidth]{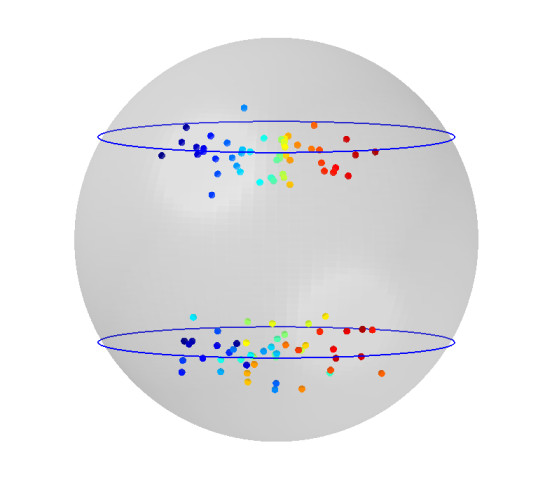}
		\caption{MS2 on $(\mathbb{S}^2)^2$}
		\label{fig:S2.dependent}
	\end{subfigure}
    \caption{Random samples from the S2, iMS2 and MS2 distributions. Same colors represent same observations. (a) low concentrations ($\kappa_0 = 10$, $\kappa_1 = 1$). (b)  independent directions with high concentrations ($\kappa_{0(k)} = 100$, $\kappa_{1(k)} = 10$, $k = 1,2,3$). (c)  horizontally dependent directions with high concentrations ($\kappa_{0(k)} = 50$, $\kappa_{1(k)} = 30$) and high dependence ($\lambda_{12} = 24$).} \label{fig:random.S2}
	\end{center}
\end{figure}

\section{Estimation}\label{sec:estimation}

%
%
%

\subsection{Approximate maximum likelihood estimation for  S1 and iMS1}\label{sec:estimation.S1}

The standard way to estimate parameters of the S1 is to use the maximum likelihood estimates (m.l.e). However, it does not seem possible to obtain explicit expressions of the m.l.e., partly due to having no closed-form expression of the normalizing constant (\ref{eq:S1density}). We propose to use an approximate m.l.e., obtained by iterating between updates for $\mu_0$ and for $(\mu_1, \kappa_0, \kappa_1)$, where each separate problem can be solved efficiently. The proposed estimation procedure, given below, only needs to specify an initial value for $\nu = \mu_0^\top \mu_1 \in (-1,1)$.
Our procedure naturally extends to estimation for the iMS1 distribution, which will also be discussed.

As a preparation, we first describe an approximation of the normalizing constant.
%
%
%
%
For this, 
we follow \cite{Kume2005}, who used saddle-point density approximations for approximating  normalizing constants of Fisher-Bingham distributions. 
The normalizing constant of the S1  has an alternative expression, as shown in the following.
\begin{prop}\label{prop:S1constant}
 For any $h>0$, let $\xi = (\frac{\nu(2\kappa_0+\kappa_1)}{2(\kappa_0+h)}, \frac{\kappa_1\sqrt{1-\nu^2}}{2h},0,\ldots,0)^\top \in \Real^p$ and let $\Psi$ be the $p\times p $ diagonal matrix with diagonal elements $(\kappa_0+h, h,\ldots,h)$. Moreover, let $g(r)$ $(r>0)$ be the probability density function of $R = Z^\top Z$, where $Z \sim N_p(\xi, \half\Psi^{-1})$. Then the normalizing constant $a(\kappa_0,\kappa_1,\nu)$ of (\ref{eq:S1density}) is
  \begin{equation}\label{eq:norm.const.FB2}
    a(\kappa_0,\kappa_1,\nu) = 2\pi^{p/2}|\Psiv|^{-1/2} g(1)\exp\left( \xi^T \Psi \xi +  h - \kappa_0 \nu^2  \right).
\end{equation}
\end{prop}

 In Proposition~\ref{prop:S1constant}, the function $g$ is the density of a linear combination of independent noncentral $\chi^2_1$ random variables. 
 Following \cite{Kume2005}, we use saddle-point density approximations in the numerical computation of $g(1)$. First, note that the derivatives of the cumulant generating function, $K_g(t) = \log \int_0^\infty e^{tr} g(r)dr$, associated with the density $g$ have closed-form expressions. Denoting by $K_g^{(j)}(t)$ the $j$th derivative of $K_g(t)$, for $j = 1,\ldots,4$, we get
  \begin{equation*}
    K_g^{(j)}(t) = \frac{(j-1)!}{2}\left( \frac{1}{(\kappa_0 + h - t)^j} + \frac{p-1}{(h-t)^j} \right) + \frac{j!}{4} \left( \frac{\nu^2(2\kappa_0 + \kappa_1)^2}{(\kappa_0 + h - t)^{j+1}} + \frac{\kappa_1^2(1 - \nu^2)}{(h - t)^{j+1}} \right).
\end{equation*}

Let $\hat{t}$ be the unique solution in $(-\infty, h)$ of the saddle-point equation $K_g^{(1)}(t) = 1$, which can be easily evaluated by using, e.g., a bisection method. Then a saddle-point density approximation of $g(1)$ is
  \begin{equation}\label{eq:norm.const.approx1}
\hat{g}(1) =  (2\pi K_g^{(2)}(\hat{t}))^{-1/2}\exp(K_g(\hat{t}) - \hat{t}+T),
\end{equation}
where
$    T =  {K_g^{(4)}(\hat{t})}/\{ 8 (K_g^{(2)}(\hat{t}))^2\} -   { 5 (K_g^{(3)}(\hat{t}))^2 }/ \{24{(K_g^{(2)}(\hat{t}))^{3}} \}$.
In the following, we approximate the value of $a(\kappa_0,\kappa_1,\nu)$ by $\hat{a}(\kappa_0,\kappa_1,\nu)$ obtained by plugging (\ref{eq:norm.const.approx1}) in place of $g(1)$ in (\ref{eq:norm.const.FB2}).

We are now ready to describe our estimation procedure. Suppose $x_1,\ldots,x_n$ is a sample from $\textrm{S1}(\mu_0,\mu_1,\kappa_0,\kappa_1)$. Given an initial value  $\hat\nu^{(0)} \in (-1,1)$, we iterate between steps 1 and 2 below until convergence.

\textbf{Step 1}: Updating $\mu_0$ given all other parameters, of which we only need $\nu$.
Suppose the inner product $\nu = \mu_0^\top\mu_1 \in (-1,1)$ is fixed. Then maximizing the likelihood function with respect to $\mu_0$ is equivalent to minimizing $\frac{1}{n} \sum^n_{i=1}(\mu_0^\top x_i - \nu)^2$ subject to the constraint $\mu_0^\top\mu_0 = 1$. With a Lagrangian multiplier $\lambda$ using matrix notation, we solve
\begin{equation}\label{eq:object.mu0.matrix}
    \min_{\mu_0 \in \mathbb{S}^{p-1}} \left[ \frac{1}{n} \|\mathbb{X}^\top\mu_0 - \nu 1_n\|^2 - \lambda(\mu_0^\top\mu_0 - 1) \right],
\end{equation}
where $\mathbb{X}$ is the $p \times n$ matrix whose $i$th column is $x_i$, yielding the necessary condition
$
    S\mu_0 - \nu\bar{x} - \lambda\mu_0 = 0,
$
where $S = \mathbb{X}\mathbb{X}^\top/n$, $\bar{x} = \frac{1}{n}\sum^n_{i=1}x_i$. For a fixed Lagrangian multiplier $\lambda$, the solution is $\hat{\mu}_0 = \nu(S - \lambda I_p)^{-1}\bar{x}$, provided that $S$ is of full rank. The constraint $\mu_0^\top\mu_0 = 1$ makes us find a root $\lambda$ of $\nu^2\bar{x}^\top(S - \lambda I_p)^{-2}\bar{x} - 1$. The root $\hat{\lambda}$ is found by a bisection search in the range $[-\nu^2\bar{x}^\top\bar{x}, \lambda_S]$, where $\lambda_S > 0$ is the smallest eigenvalue of $S$ \citep{Browne1967}. The solution to (\ref{eq:object.mu0.matrix}) is then
\begin{equation}\label{eq:step1_s1_mu0}
    \hat{\mu}_0 = \nu(S - \hat{\lambda}I_p)^{-1}\bar{x}.
\end{equation}
If $\nu = 0$, then $\hat{\mu}_0$ is the eigenvector of $S$ corresponding to the smallest eigenvalue.

\textbf{Step 2}: Updating $(\mu_1, \kappa_0, \kappa_1)$ given $\mu_0$. To facilitate the estimation of $\mu_1$, let
$E(X) \in \Real^p$ be the usual expected value of the random vector $X$ in the ambient space $\Real^p$. We use the fact that $\mu_1$ is a linear combination of $\mu_0$ and $\gamma_0 := E(X)/\|E(X)\|$ (due to Lemma~\ref{lem:ambientMeanLemma} in the Appendix)
and reparameterize $\mu_1$ by an angle $\varphi$ and the direction ${\gamma^*} = P_{\mu_0} \gamma_0 / \norm{P_{\mu_0} \gamma_0}$ orthogonal to $\mu_0$,
giving
\begin{equation}\label{eq:mu1.for.S1}
\mu_1 = \cos(\varphi) \mu_0 + \sin(\varphi) \gamma^*.
\end{equation}
With $\gamma_0$ estimated by $\hat\gamma_0 = \bar{x} / \norm{\bar{x}}$, we now optimize for $\varphi \in [0,2\pi)$ together with $\kappa_0, \kappa_1$, as follows. With $\mu_0$ and $\hat\gamma_0$ (thus $\hat\gamma^*$) given, the approximate negative log-likelihood with respect to $(\varphi, \kappa_0, \kappa_1)$ is
\begin{equation}\label{eq:approxlikeli}
    \tilde{\ell}_{\mu_0}(\varphi, \kappa_0, \kappa_1) = -n\log{\hat{a}(\kappa_0, \kappa_1, \varphi)} - \kappa_0\sum_{i=1}^n(\mu_0^\top x_i - \cos\varphi)^2 + \kappa_1\sum_{i=1}^n x_i^\top( \cos(\varphi) \mu_0 + \sin(\varphi) \hat\gamma^*).
\end{equation}
Numerically minimizing (\ref{eq:approxlikeli}) is much simpler than optimizing for $\mu_1$ with the nonlinear constraint $\norm{\mu_1} = 1$. We use a standard optimization package to obtain $\hat\varphi, \hat\kappa_0, \hat\kappa_1$ that minimizes (\ref{eq:approxlikeli}). We get $\hat\mu_1$ by substituting $(\gamma^*,\varphi)$ by $(\hat\gamma^*,\hat\varphi)$ in (\ref{eq:mu1.for.S1}) and $\hat{\nu} = \cos(\hat\varphi)$.

%

Let us now describe  an extension of the above algorithm to the iMS1 model. Suppose $(x_{i(1)},\ldots,x_{i(K)}) \in (S^{p-1})^K$ for $i = 1,\ldots,n$ is a sample from an iMS1 model, where each marginal distribution is  $\textrm{S1}(\mu_0 ,\mu_{1(j)} ,\kappa_{0(j)} ,\kappa_{1(j)} )$. While Step 2 above can be applied to update $\mu_{1(k)} ,\kappa_{0(k)} ,\kappa_{1(k)}$ given $\mu_0$, we modify Step 1 by replacing (\ref{eq:object.mu0.matrix}) with
\[
    \min_{\mu_0 \in \mathbb{S}^{p-1}} \left[\frac{1}{n}\sum_{j=1}^K (\kappa_{0(j)} \| \mathbb{X}_{(j)}^\top \mu_0 - \nu_{(j)} 1_n \|^2) - \lambda(\mu_0^\top\mu_0 - 1) \right],
\]
where the marginal $p\times n$ observation matrix $\mathbb{X}_{(j)}$ has the columns $x_{i(j)}$ $(i = 1,\ldots,n)$. This is solved with the obvious analog  to (\ref{eq:step1_s1_mu0}).


\subsection{Estimation via profile likelihood for S2, iMS2 and MS2}\label{sec:estimation.S2}
The S2 model and its extensions have the convenient property that the horizontal components are independent of the vertical ones.  To take advantage of this, suppose for now that $\mu_0$ is known. This allows to decompose an observation $x$ into two independent random variables $(s,y)$, which in turn leads to an easy estimation of the remaining parameters $\eta:= (\mu_1,\kappa_0,\kappa_1)$.
Thus our strategy of computing the m.l.e. proceeds in two nested steps. Let $\ell_n(\mu_0,\eta)$ be the negative likelihood function given a sample $x_1,\ldots,x_n$ from S2$(\mu_1,\eta)$.
In the \emph{outer step}, we update $\mu_0$ to maximize a profile likelihood, i.e.,
\begin{equation}\label{eq:S2est.outer}
\hat\mu_0 = \argmin_{\mu_0} \ell_n(\mu_0,\hat\eta_{\mu_0}),
\end{equation}
where evaluating
\begin{equation}\label{eq:S2est.inner}
\hat\eta_{\mu_0} = \argmin_{\eta} \ell_n(\mu_0,\eta)
\end{equation}
for a fixed $\mu_0$ is the \emph{inner step}.
It is straightforward to see that the m.l.e. of $(\mu_0,\eta)$ is given by $(\hat\mu_0, \hat\eta_{\hat\mu_0})$.

In the following, we discuss in detail the inner step (\ref{eq:S2est.inner}) of minimizing $\ell_{\mu_0}(\eta) := \ell_n(\mu_0,\eta)$ for the iMS2 model (\ref{eq:indep.mv.S2}) and for the MS2 model (\ref{eq:dep.str.mult.S2density}), while we resort to  a standard optimization package for solving (\ref{eq:S2est.outer}).  

\textbf{Independent multivariate S2 (iMS2).} Suppose $(x_{i(1)},\ldots,x_{i(K)}) \in (S^{p-1})^K$ for $i = 1,\ldots,n$ is a sample from an iMS2, where each marginal distribution is  S2$(\mu_0 ,\mu_{1(j)} ,\kappa_{0(j)} ,\kappa_{1(j)} )$. For a given $\mu_0$, the joint density can be written in terms of $(\sv_i,\bphi_i)$ as done in (\ref{eq:dep.str.mult.S2density}), but with $\bLambda = 0$. Furthermore, with $ \bkappa_0 = (\kappa_{0(1)},\ldots,\kappa_{0(K)})^\top $, $\nuv= (\nu_{(1)},\ldots,\nu_{(K)})^\top$, we can write
 $$ H^\top\textbf{s}_i - \textbf{s}_i^\top D_{\kappa_0} \textbf{s}_i = - (\boldsymbol{\kappa}_0 \circ (\textbf{s}_i - \boldsymbol{\nu}))^\top (\textbf{s}_i - \boldsymbol{\nu}),$$
 where $\circ$ denotes the element-wise product, and hence
\begin{equation}\label{eq:S2.Ttob}
\log \left[ T_1(H,D_{\kappa_0})T_3(\bkappa_1, 0) \right]= \sum_{j=1}^K \left[ \log b(\kappa_{0(j)} ,\kappa_{1(j)}, \nu_{(j)}) + \kappa_{0(j)} \nu_{(j)}^2 \right].
\end{equation}
 Note that the normalizing constant $b(\kappa_0,\kappa_1,\nu)$   satisfies
\begin{eqnarray*}\label{eq:S2constant}
	b(\kappa_0, \kappa_1, \nu) 
	&=&  \int_{0}^{2\pi} e^{\kappa_1 \cos\phi} d\phi \int_{-1}^{1} e^{-\kappa_0 (s - \nu)^2} ds  \nonumber \\
	&=&  (2\pi)^{3/2} (2\kappa_0)^{-1/2}\mathcal{I}_0(\kappa_1) \left[ \Phi((1-\nu)\sqrt{2\kappa_0}) - \Phi(-(1+\nu)\sqrt{2\kappa_0}) \right],
\end{eqnarray*}
where $\Phi(\cdot)$ is the  standard normal distribution function. Finally, the negative log-likelihood function (given $\mu_0$) is
\begin{equation}
\ell_{\mu_0}(\boldsymbol{\nu}, \boldsymbol{\zeta}, \boldsymbol{\kappa}_0, \boldsymbol{\kappa}_1; \{\textbf{s}_i, \boldsymbol{\phi}_i\}_{i=1}^n) = \ell_{\mu_0}^{(1)}(\boldsymbol{\nu}, \boldsymbol{\kappa}_0) + \ell_{\mu_0}^{(2)}(\boldsymbol{\zeta}, \boldsymbol{\kappa}_1), \label{eq:S2.nll.canonical}
\end{equation}
where
\begin{eqnarray}
\ell_{\mu_0}^{(1)}(\boldsymbol{\nu}, \boldsymbol{\kappa}_0) &=& \sum_{j=1}^K \left[ \kappa_{0(j)}\sum_{i=1}^n(s_{i(j)} - \nu_{(j)})^2 -\frac{n}{2}\log(2\kappa_{0(j)}) + \frac{n}{2}\log(2\pi) \right. \label{eq:S2.nll.canonical.1} \\
    && \left. + n\log\left( \Phi\left( (1-\nu_{(j)})\sqrt{2\kappa_{0(j)}} \right) - \Phi\left( -(1+\nu_{(j)})\sqrt{2\kappa_{0(j)}}\right) \right) \right], \nonumber \\
    \ell_{\mu_0}^{(2)}(\boldsymbol{\zeta}, \boldsymbol{\kappa}_1) &=& - \sum_{j=1}^K \left[ \kappa_{1(j)}\sum_{i=1}^n\cos(\phi_{i(j)} - \zeta_{(j)}) - n\log \mathcal{I}_0(\kappa_{1(j)}) - n\log(2\pi) \right]. \label{eq:S2.nll.canonical.2}
\end{eqnarray}
Therefore, the optimization for the inner step (\ref{eq:S2est.inner}) is equivalent to simultaneously solving $2K$ subproblems.

Each of the $K$ subproblems of (\ref{eq:S2.nll.canonical.1}) is equivalent to obtaining the m.l.e. of a truncated normal distribution $\textrm{trN}(\nu_{(j)}, (2\kappa_{0(j)})^{-1/2}; (-1,1))$
based on the observations $s_{i(j)}$ ($i = 1,\ldots,n$).
Similarly, each of the $K$ subproblems of (\ref{eq:S2.nll.canonical.2}) amounts to obtaining the m.l.e. of a von Mises distribution with mean $\zeta_{(j)}$ and concentration $\kappa_{1(j)}$ from the sample $\phi_{i(j)}$ ($i = 1,\ldots,n$). The m.l.e.s of the truncated normal are numerically computed, and we use the method of \cite{Banerjee2005} to obtain approximations of the m.l.e.s of the von Mises.


\textbf{MS2.}
Under the general $\textrm{MS2}$ model (\ref{eq:dep.str.mult.S2density}) with a dependence structure on $\bphi_i$, a decomposition
$
\ell_{\mu_0}(\boldsymbol{\nu}, \boldsymbol{\zeta}, \boldsymbol{\kappa}_0, \boldsymbol{\kappa}_1, \bLambda) = \ell_{\mu_0}^{(1)}(\boldsymbol{\nu}, \boldsymbol{\kappa}_0) + \ell_{\mu_0}^{(2)}(\boldsymbol{\zeta}, \boldsymbol{\kappa}_1,\bLambda)$, similar to (\ref{eq:S2.nll.canonical}), is valid, where (\ref{eq:S2.nll.canonical.2}) is replaced by

\begin{equation}\label{eq:ell2}
    \ell_{\mu_0}^{(2)}(\boldsymbol{\zeta}, \boldsymbol{\kappa}_1, \boldsymbol{\Lambda}) = - \sum_{i=1}^n \left[ \boldsymbol{\kappa}_1^\top c(\boldsymbol{\phi}_i,\boldsymbol{\zeta}) + \frac{1}{2}s(\boldsymbol{\phi}_i, \boldsymbol{\zeta})^\top \boldsymbol{\Lambda} s(\boldsymbol{\phi}_i, \boldsymbol{\zeta}) - \log T_3(\boldsymbol{\kappa}_1,\boldsymbol{\Lambda}) \right].
\end{equation}
Minimizing (\ref{eq:ell2}) is equivalent to computing the m.l.e. of the  multivariate von Mises distribution \citep{Mardia2008}. We either use maximum pseudo-likelihood estimate as discussed in \cite{Mardia2008} or moment estimates, yielding
\begin{equation}\label{eq:mvMF.MoM}
    \hat{\zeta}_{(j)} = \frac{1}{n}\sum_{i=1}^n \phi_{i(j)} / \| \frac{1}{n}\sum_{i=1}^n \phi_{i(j)} \| ,\quad \hat{\kappa}_{1(j)} = \bar{S}^{-1}_{jj}, \quad \hat{\lambda}_{(jk)} = \bar{S}^{-1}_{jk} \quad (j \neq k),
\end{equation}
where $\bar{S} = (\bar{S}_{jk})$ and $\bar{S}_{jk} = \frac{1}{n}\sum_{i=1}^n \sin(\phi_{i(j)} - \hat{\zeta}_{(j)})\sin(\phi_{i(k)} - \hat{\zeta}_{(k)})$ for $j,k  = 1,\ldots,K$. These estimates coincide with the m.l.e.s when $K=2$. For larger $K>3$, the accuracy of the moment estimates deteriorates, but evaluating m.l.e.s or a maximum pseudo-likelihood estimator becomes computationally highly expensive.

\section{Testing hypotheses}\label{sec:testing}
It is of interest to infer on the parameters of our models. In this section, we describe a large-sample testing procedure for several hypotheses of interest.

Our testing procedure is based on the likelihood ratio statistic, with effort devoted to an identification of the restricted parameter space $\Theta_0$, for each hypothesis and computing the maximized likelihood under $\Theta_0$.
Recall that the parameter space for the iMS1 and iMS2 models is given by
$\Theta_{\rm ind} = \mathbb{S}^{p-1} \times (\mathbb{S}^{p-1})^{K} \times (\Real_+)^{K} \times (\Real_+)^{K}$ for $\theta_{\rm ind} = (\mu_0, \muv_1, \kappav_0, \kappav_1)$.
For the more general GMS2 model including associations, we have $\Theta_{\rm GMS2} = \Theta_{\rm ind} \times ( \Real^{ (p-1)^2 } )^{K(K-1)/2}$ for $\theta_{\rm GMS2} = (\theta_{\rm ind}, \Bv)$. In the following, we describe our testing procedure using the MS2 distribution in dimension $p =3$, whose parameter space is $\Theta = (\theta_{\rm ind} \times (\Real)^{K(K-1)/2})$ for $\theta = (\theta_{\rm ind}, \Lambdav)$.
For some $\Theta_0$ that dictates a null hypothesis $H_0$ and satisfies $\Theta_0 \subset \Theta$, we denote the maximized log-likelihood under $\Theta_0$ by $\mathcal{L}_0$, and the maximized log-likelihood under $\Theta$ by $\mathcal{L}_1$. It is well-known that for this nested model, for large sample size $n$, $W_n := -2(\mathcal{L}_0 - \mathcal{L}_1)$ follows approximately a chi-square distribution with $q_1-q_2$ degrees of freedom, where $q_1 = (p-1)(K+1) + 2K + K(K-1)/2$ and $q_2$ are the dimensions of $\Theta$ and $\Theta_0$, respectively. Once $W_n$ is computed, as usual, our test rejects the null hypothesis for large enough values of $W_n$.

We are interested in the following null hypotheses, with the alternative being the full MS2 distribution. Testing the first hypothesis gives a test of association among directional vectors, while the latter two provide tests regarding the parameters of the underlying small-sphere  $\Cc(\muv_0,\nuv)$. In all three cases, the alternative is $H_1: \theta \in \Theta \setminus \Theta_0$.

1. \textbf{Test of association}. $H_0$: $\bLambda = \textbf{0}$, i.e., $\theta \in \Theta_0 = \mathbb{S}^{p-1} \times (\mathbb{S}^{p-1})^{K} \times (\Real_+)^{K} \times (\Real_+)^{K} \times \{ \textbf{0} \}$. Under $H_0$, the model degenerates to the iMS2 and there is no horizontal dependence.

2. \textbf{Test of axis}. $H_0$: $\mu_0 = \mu_0^*$, i.e., $\theta \in \Theta_0 = \{\mu_0^*\} \times (\mathbb{S}^{p-1})^{K} \times (\Real_+)^{K} \times (\Real_+)^{K} \times (\Real)^{K(K-1)/2}$.  This is to test whether a predetermined axis $\mu_0^*$ of the small sphere is acceptable.

3. \textbf{Test of great-sphere}. $H_0$: $\nuv = 0$, i.e., $\theta \in \Theta_0 \simeq \mathbb{S}^{p-1} \times (\mathbb{S}^{p-2})^{K} \times (\Real_+)^{K} \times (\Real_+)^{K} \times (\Real)^{K(K-1)/2}$. ($A \simeq B$ means that $A$ and $B$ are diffeomorphic.) This is to test whether the underlying spheres are great spheres with radius 1.

While the test of association (Hypothesis 1) is only available under the MS2 model ($p=3$), Hypotheses 2 and 3 can also be tested using S1, S2, iMS1 or iMS2 models in any dimension $p\geq 3$.
Moreover, to validate the use of small-sphere distributions, in any dimension $p\geq 3$, the following hypotheses can be tested. For simplicity, assume for now the S1 model with $\theta = (\mu_0,\mu_1,\kappa_0,\kappa_1) \in \Theta = (\mathbb{S}^{p-1})^2 \times (\Real_+)^2$.

4. \textbf{Test for von Mises-Fisher distribution}. $H_0$: $\kappa_0 = 0$, i.e., $\theta \in \Theta_0 \simeq \mathbb{S}^{p-1} \times \Real_+$. Under $H_0$, there is no ``small-circle feature.''

5. \textbf{Test for Bingham-Mardia distribution}. $H_0$: $\kappa_1 = 0$, i.e., $\theta \in \Theta_0 \simeq \mathbb{S}^{p-1} \times \Real_+$. Under $H_0$, there is no unique mode.

Hypotheses 4 and 5 can also be tested under the S2, iMS1, iMS2 and MS2 models. Additional care is needed when using S2, iMS2, or MS2 for Hypothesis 4, as the null distribution is the von Mises-Fisher on $\Cc(\muv_0,\nuv)$ (not on $\mathbb{S}^{p-1}$).

For each hypothesis, computing the test statistic $W_n$ requires to maximize the likelihood on $\Theta_0$ (or to compute $\Lc_0$). This is easily achieved by modifying the iterative algorithms in Section~\ref{sec:estimation}.
For example, for the test of association, computing $\Lc_0$ and $\Lc_1$ amounts to obtaining the m.l.e.s under the iMS2 and MS2 models, respectively; for Hypothesis 2 (test of axis), where $\mu_0$ is given, one only needs to solve (\ref{eq:S2est.inner}) once. Other cases of restricted m.l.e.s can be easily obtained.
In the online supplementary material, we confirm that the test statistic $W_n$ using our algorithms under the null hypotheses above are empirically nearly chi-square distributed for sample size $n = 30$. In Section~\ref{sec:isotropic example} and in the online supplementary material, empirical powers of the proposed test procedures are reported for several important alternatives.

\section{Numerical studies}\label{sec:simulation}


We demonstrate the performances of small-circle fitting in Section~\ref{sec:simulation.smallcircle},
the ability of the MS2 of modeling the horizontal dependence in Section~\ref{sec:simulation.association} and a testing procedure to prevent overfitting in Section~\ref{sec:isotropic example}.

\subsection{Estimation of small-circles}\label{sec:simulation.smallcircle}


The performance of our estimators in fitting the underlying small-spheres $\mathcal{C}(\mu_0,\nu)$ is numerically compared with those of competing estimators obtained from assuming the Bingham-Mardia (BM) distribution and the least-square estimates of \cite{Schulz2015}.
The BM distribution has originally been defined only for data on $\mathbb{S}^2$, but we use a natural extension given by a special case of the iMS1. Thus, ``BM estimates'' refer to the estimates of the iMS1 model with the restriction $\bkappa_1 = 0$. The estimates of Schulz et al. are obtained by minimizing the sum of squared angular distances from observations to $\Cc(\hat\mu_0,\hat\nu)$, which will be referred to as a ``least-squares (LS)'' method.

We first consider four univariate S2 models to simulate data concentrated on a small circle. The directional parameters $(\mu_0,\mu_1)$ are set to satisfy $\nu = 0.5$.
We use $(\kappa_0, \kappa_1) = (10, 1), (100, 1), (100, 10)$ to represent various data situations. Random samples from these three settings are shown in Fig.~\ref{fig:examples.simulation}(a)--(c).
We also consider the BM model as a special case of the S2 distributions (by setting $\kappa_1 = 0$); a sample from the BM distribution is shown in Fig.~\ref{fig:examples.simulation}(d).

\begin{figure}[t!]
	\begin{center}
	\begin{subfigure}{0.24\textwidth}
		\includegraphics[width = \textwidth]{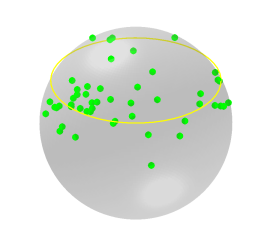}
        \caption{$\kappa_0 = 10, \kappa_1 = 1$}
	\end{subfigure}
	\begin{subfigure}{0.24\textwidth}
		\includegraphics[width = \textwidth]{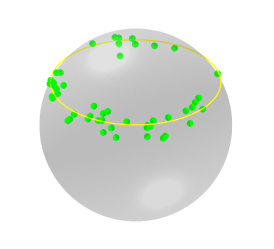}
        \caption{$\kappa_0 = 100, \kappa_1 = 1$}
	\end{subfigure}
	\begin{subfigure}{0.24\textwidth}
		\includegraphics[width = \textwidth]{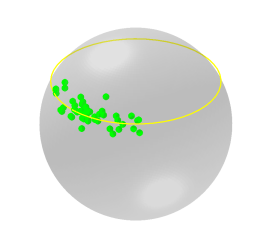}
        \caption{$\kappa_0 = 100, \kappa_1 = 10$}
	\end{subfigure}
	\begin{subfigure}{0.24\textwidth}
		\includegraphics[width = \textwidth]{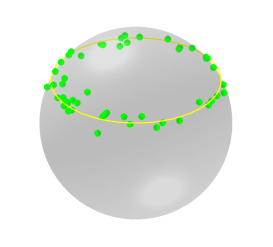}
        \caption{$\kappa_0 = 100, \kappa_1 = 0$}
	\end{subfigure}
    \caption{Random samples of size $n = 50$ from the S2 model on $\mathbb{S}^2$ used in our simulations. Small-circle estimation performances are reported in Table~\ref{tab:simul.S2.K1}.}\label{fig:examples.simulation}
    \end{center}
\end{figure}

The small-circle estimation performances of the S1, S2, BM, and LS estimates are measured by an angular product error (in degrees), defined as
\begin{equation}\label{eq:loss}
L\left((\mu_0,\nu), (\hat\mu_0,\hat\nu)\right) = \left(\mbox{Angle}(\hat\mu_0,\mu_0)^2 + \left( \frac{180}{\pi} (\arccos\hat\nu - \arccos\nu )\right)^2\right)^{1/2}.
\end{equation}
Table \ref{tab:simul.S2.K1} displays the means and standard deviations of ${L}\left((\mu_0,\nu), (\hat\mu_0,\hat\nu)\right)$ from 100 repetitions for each of the four methods, fitted to random samples of size 50 from each of the settings, labeled (a)--(d). The S2 estimates performed best for all non-trivial settings: (a), (b) and (c). Estimators under Settings (a) and (c) show higher standard errors. This is expected and due to either the large vertical dispersion or the smaller horizontal dispersion as visualized in Fig.~\ref{fig:examples.simulation}(a) and (c). Even when the sophisticated S1 and S2 models are not needed in Setting (d), the S1 and S2 estimators perform virtually as good as the BM estimator does.

\begin{table}[t!]
    $$\begin{array}{c|cccc}
           \mbox{Method}     & \mbox{(a)} & \mbox{(b)} & \mbox{(c)} & \mbox{(d)} \\
    \hline
    \textrm{S1} & 6.62(3.44) & 1.59(0.78) & 14.89(13.00) & \textbf{1.32}(0.56) \\
    \textrm{S2} & \textbf{6.06}(3.21) & \textbf{1.58}(0.76) & \textbf{14.57}(11.56) & 1.33(0.56) \\
    \textrm{BM} & 9.54(9.80) & 1.66(0.81) & 16.59(13.57) & \textbf{1.32}(0.56) \\
    \textrm{LS} & 6.48(3.37) & 1.61(0.75) & 14.68(11.54) & 1.33(0.55) \\
    \end{array}$$
    \caption{Small-circle estimation performances for univariate data on $\mathbb{S}^2$ from Fig.~\ref{fig:examples.simulation}. Means (standard deviations) of the angular product errors in degrees (\ref{eq:loss}) are shown.}\label{tab:simul.S2.K1}
\end{table}

Next, to show the performance of our multivariate models, we consider six bivariate MS2 models. The directional parameters $(\mu_0, \muv_1)$ were set to satisfy $\nuv = (0.5, -.3)$, and the concentration parameters were chosen to mimic the concentrations of the univariate models, described above. For Cases (a)--(c), we set $(\kappa_{0j},\kappa_{1j},\lambda_{12}) = (10, 1, 0), (100, 1, 0), (100, 10, 0)$, for $j = 1,2$, so that the models are indeed the iMS2. For the later three cases (d)--(f), we set $(\kappa_{0j},\kappa_{1j},\lambda_{12}) = (10, 2, 1.5), (100, 2, 1.5), (100, 20, 15)$, $j = 1,2$, to make their vertical and horizontal dispersions be similar to the iMS2 counterparts. By setting $\lambda_{12}>0$, the random bivariate directions are positively associated. (Examples of random samples from these settings can be found in Fig.~A6 in the online supplementary material.)
%
%
%
%
The small-circle estimation performance of the iMS1, iMS2, MS2, BM and LS estimates is measured by the canonical multivariate extension of the angular product error (\ref{eq:loss}). Table \ref{tab:simul.S2.K2} collects the means and standard deviations of the angular product errors 
from 100 repetitions
with the sample size $n = 50$.
Overall, the three proposed models (iMS1, iMS2, and MS2) show better or competitive performances in the axis and radii estimation. In particular, when directions are clearly concentrated on small-circles and are horizontally dependent, i.e., in Settings (e) and (f), the MS2 estimates shows better performances than others.

We check robustness against model misspecification of the estimators  by simulating data from a more general signal-plus-noise model
 (neither S1 nor S2). The performances of small-circle fitting of the proposed methods are comparable to that of the least-square estimator. Relevant simulation results and a detailed discussion can be found in the online supplementary material.

\begin{table}[t!]
\begin{center}
	$$\begin{array}{c|ccc|ccc}
    & \multicolumn{3}{c}{\mbox{Independent}} & \multicolumn{3}{c}{\mbox{Dependent}} \\
    \mbox{Method} & \mbox{(a)} & \mbox{(b)} & \mbox{(c)} & \mbox{(d)} & \mbox{(e)} & \mbox{(f)} \\
    \hline
    \textrm{iMS1} & 4.52(1.89) & 1.28(0.51) & \textbf{3.90(3.61)} & 7.01(2.88) & 1.58(0.76) & 4.71(5.39) \\
    \textrm{iMS2} & 4.45(1.71) & \textbf{1.27(0.51)} & 4.30(2.46) & \textbf{5.78(2.50)} & 1.58(0.75) & 4.60(2.69) \\
    \textrm{MS2} & \textbf{4.40(1.71)} & 1.28(0.51) & 4.26(2.45) & 5.90(2.46) & \textbf{1.57(0.75)} & \textbf{4.49(2.79)} \\
    \textrm{BM} & 5.15(2.36) & 1.28(0.52) & 8.63(4.66) & 17.02(21.46) & 1.69(0.83) & 10.77(5.64) \\
    \textrm{LS} & 4.48(1.86) & 1.28(0.53) & 4.30(2.35) & 6.81(3.20) & 1.59(0.71) & 4.62(2.82) \\
	\end{array}$$
    \caption{ Small-circles estimation performances for bivariate data on $\mathbb{S}^2$ from Fig.~A6 in the online supplementary material. Means (standard deviations) of the angular product errors in degrees (\ref{eq:loss}) are shown. }\label{tab:simul.S2.K2}
\end{center}
\end{table}

\subsection{Estimation of horizontal dependence}\label{sec:simulation.association}
The ability of the MS2 to model the horizontal dependence is an important feature of the proposed distributions. Here, we empirically confirm that the MS2 estimates provide accurate measures of horizontal dependence, using Cases (c) and (f) in Section~\ref{sec:simulation.smallcircle}. For sample sizes $n = 50$ and $200$, the concentration and association parameters were estimated under the assumption of MS2 (or iMS2), and  Table~\ref{tab:compare.models.concentration} summarizes the estimation accuracy. In all cases, the MS2 model provides precise estimations of the horizontal dispersion and dependence; as the sample size increases, the mean squared error decreases. For Case (c), the underlying model is exactly iMS2, so the iMS2 estimates have smaller mean squared errors than the MS2 estimates.
However, for Case (f), we notice that the iMS2 estimates of $\boldsymbol{\kappa}_1 = (\kappa_{11}, \kappa_{12})$ become inferior.
In fact, in case of existing horizontal dependence, i.e., when $\lambda_{12} \neq 0$, the concentration parameters $\kappav_1$ in the misspecified iMS2 model do not correctly represent the concentrations as correctly represented by the MS2 model.
%
This is so because the marginal distribution of $\phi_j$, $j = 1,2$, in (\ref{eq:dep.str.mult.S2density}) is not a von Mises distribution \citep{Shing2002, Mardia2008}.

\begin{table}[t!]
\begin{center}
	$$\begin{array}{cc|ccc|ccc}
	 & & \multicolumn{3}{c}{\mbox{(c)}} & \multicolumn{3}{c}{\mbox{(f)}} \\
    n & \mbox{Method} & \kappa_{11} = 10 & \kappa_{12} = 10 & \lambda_{12} = 0 & \kappa_{11} = 20 & \kappa_{12} = 20 & \lambda_{12} = 15 \\
    \hline
    \multirow{2}{*}{50} & \textrm{iMS2} & 10.23(2.48) & 10.54(2.16) &  & 11.73(2.19) & 11.19(2.14) &  \\
     & \textrm{MS2} & 10.48(2.54) & 10.80(2.27) & -0.17(1.85) & 22.63(5.06) & 21.40(4.32) & 16.82(4.22) \\
   \multirow{2}{*}{200} & \textrm{iMS2} & 10.31(1.08) & 10.10(1.04) &  & 11.00(1.05) & 11.09(1.04) &  \\
     & \textrm{MS2} & 10.35(1.10) & 10.14(1.05) & -0.12(0.73) & 20.38(2.17) & 20.54(2.25) & 15.41(2.06) \\
	\end{array}$$
    \caption{ Concentration and association parameter estimates for bivariate data on $(\mathbb{S}^2)^2$ from Fig. A6 in the online supplmentary material. Means (standard deviations) of the estimates (from 100 repetitions). The column headings show the true parameters.}\label{tab:compare.models.concentration}
\end{center}
\end{table}

\subsection{Detecting overfitting in an isotropic case}\label{sec:isotropic example}

When the data do not exhibit a strong tendency of a small-circle feature, the S1 and S2 distributions may overfit the data. 
For example, to a random sample from an isotropic vMF distribution as shown in Fig.~\ref{fig:examples.isotropy}(a), the S1 or S2 model fits an unnecessary small-circle $\mathcal{C}(\hat\mu_0,\hat\nu)$. Indeed, a small-circle fit was observed in 83\% of simulations of fitting the S1 model. Using the BM model or the LS results in a similar overfitting, where very small circles are erroneously fitted for 100\% and 68\% of the simulations, for the BM and LS, respectively.

This problem of overfitting has been known for a while and discussed in the context of dimension reduction of directional data. In particular, \cite{Jung2011,Jung2012} and \cite{Eltzner2015} investigated the overfitting phenomenon for the least-square estimates and proposed some ad-hoc methods for adjustment.
To prevent the overfitting, we point out that the testing procedure in Section~\ref{sec:testing} for the detection of isotropic distributions (Hypothesis 4) works well. To confirm this, we evaluated the empirical power of the test at the significance level $\alpha = 0.05$ for several alternatives. The power increases sharply as the distributions become more anisotropic; under the alternative distributions depicted in Fig.~\ref{fig:examples.isotropy}(b)--(d), the empirical powers are respectively $\hat\beta = 0.435, 1$ and $1$, evaluated from 200 repetitions.

\begin{figure}[t!]
	\begin{center}
	\begin{subfigure}{0.24\textwidth}
		\includegraphics[width = \textwidth]{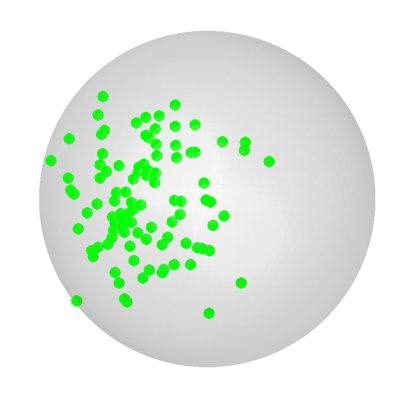}
        \caption{ vMF ($\kappa = 10$)}
	\end{subfigure}
	\begin{subfigure}{0.24\textwidth}
		\includegraphics[width = \textwidth]{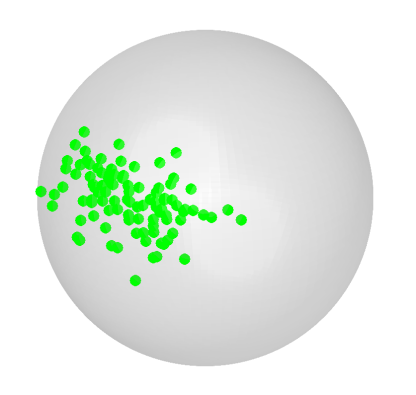}
        \caption{$(\kappa_0, \kappa_1) = (20, 10)$}
	\end{subfigure}
	\begin{subfigure}{0.24\textwidth}
		\includegraphics[width = \textwidth]{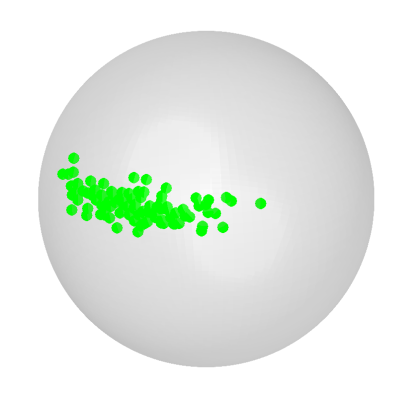}
        \caption{$(\kappa_0, \kappa_1) = (100, 10)$}
	\end{subfigure}
	\begin{subfigure}{0.24\textwidth}
		\includegraphics[width = \textwidth]{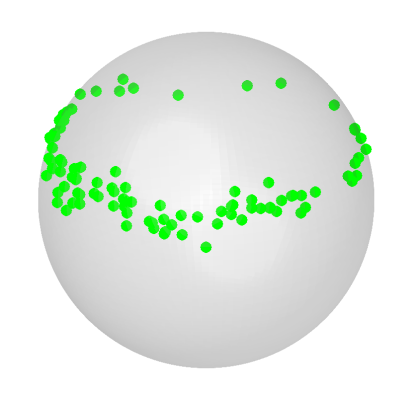}
        \caption{$(\kappa_0, \kappa_1) = (100, 1)$}
	\end{subfigure}
    \caption{Degrees of the ``small-circle feature.'' Shown are random samples from an isotropic distribution (case (a)),  and the S2 distributions with increasing ``small-circle concentrations'' (cases (b)--(d)). }\label{fig:examples.isotropy}
    \end{center}
\end{figure}

\section{Analysis of s-rep data}\label{sec:application.examples-srep}
In this section, an application of the proposed distributions and test procedures to s-rep data  is discussed.

\subsection{Modeling rotationally-deformed ellipsoids via s-reps}
Skeletal representations (s-reps) have been useful in mathematical modeling of human anatomical objects \citep{Pizer2008}. Roughly, an s-rep model for a 3-dimensional object consists of locations of a skeletal mesh (inside of the object) and spoke vectors (directions and lengths), connecting the skeletal mesh with the boundary of the object; examples are shown in the top left panel of Fig.~\ref{fig:ellipsoids-all}.
When the object is ``rotationally deformed", \cite{Schulz2015} have shown that the directional vectors of an s-rep model approximately trace a set of concentric small-circles on $\mathbb{S}^2$, as shown in the top panels of the figure.
Such rotational deformations  (e.g., rotation, bending and twisting) of human anatomical objects have been observed in between and within shape variations of hippocampi and prostates \citep{Joshi2002,Jung2011,Pizer2013}.
We demonstrate a use of the MS2 distribution in modeling (and fitting) a population of such objects via s-reps. Note that the sample space of an s-rep with $K$ spokes is $(\mathbb{S}^2)^K \times \Real_+^K \times (\Real^3)^K$ (for direction, length, and location). In this work, we choose to analyze the spoke directions in $(\mathbb{S}^2)^K$ only, leaving a full-on analysis, accommodating the lengths and locations, to future work.

\begin{figure}[t!]
	\begin{center}
    \begin{subfigure}{0.6\textwidth}
        \includegraphics[width = \textwidth]{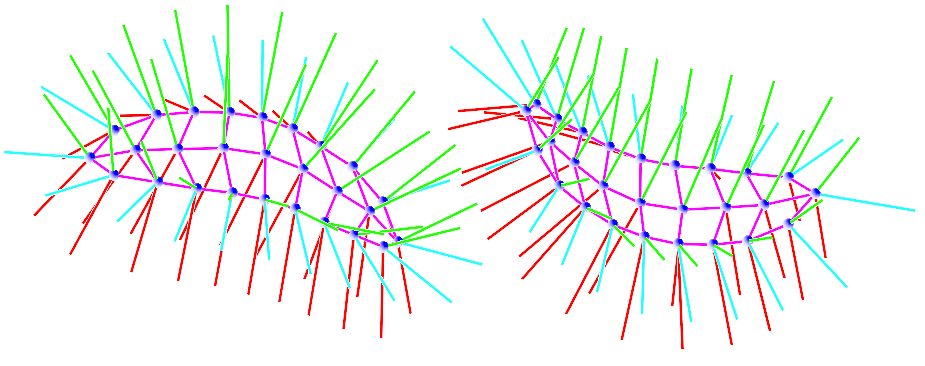}
        \caption{}
    \end{subfigure}
    \begin{subfigure}{0.35\textwidth}
        \includegraphics[width = \textwidth]{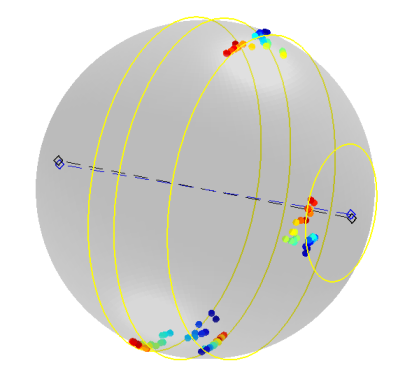} \\
        \caption{}
    \end{subfigure}
    \begin{subfigure}{0.45\textwidth}
        \includegraphics[width = \textwidth]{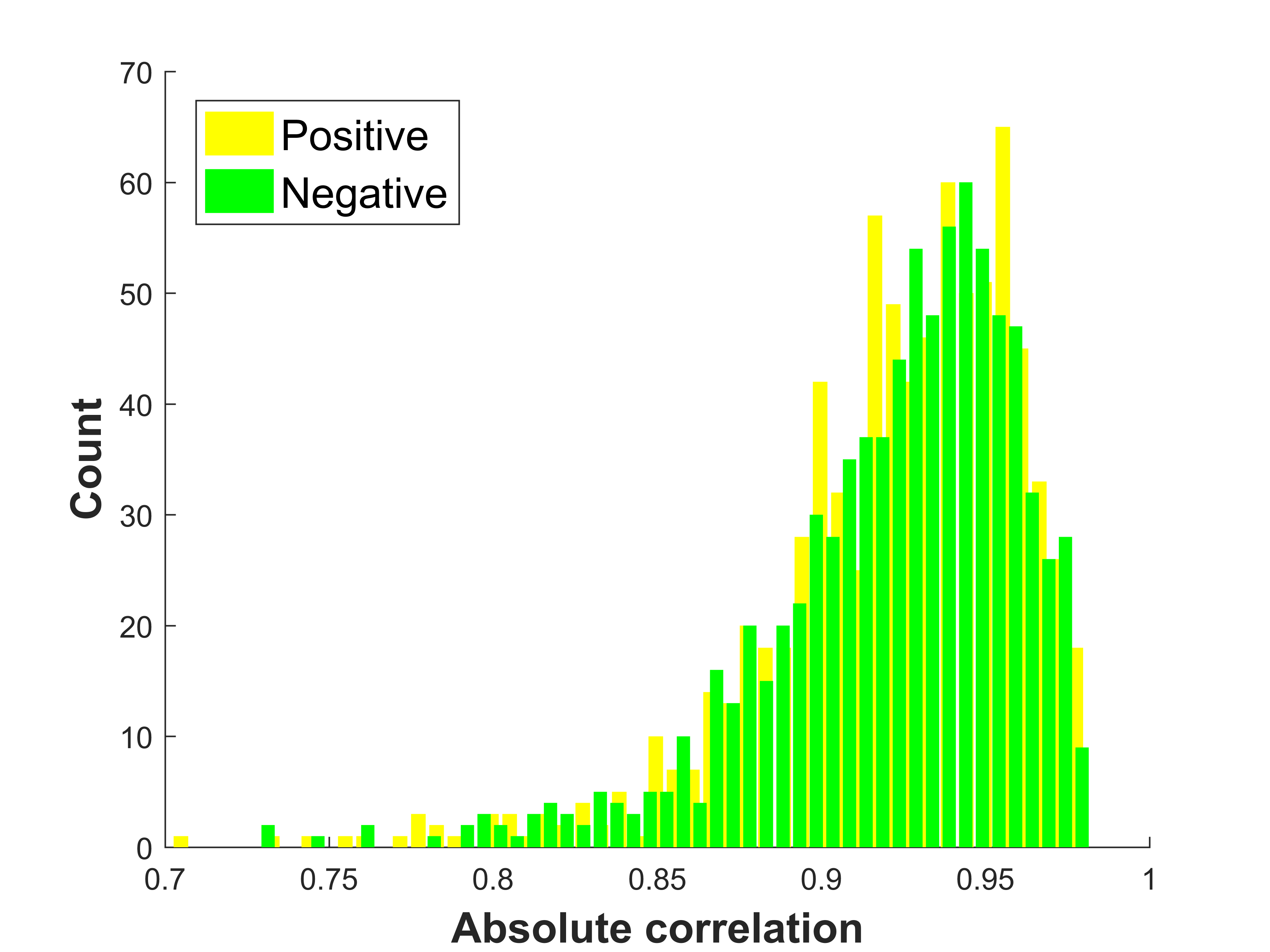}
        \caption{}
    \end{subfigure}
    \begin{subfigure}{0.45\textwidth}
        \includegraphics[width = \textwidth]{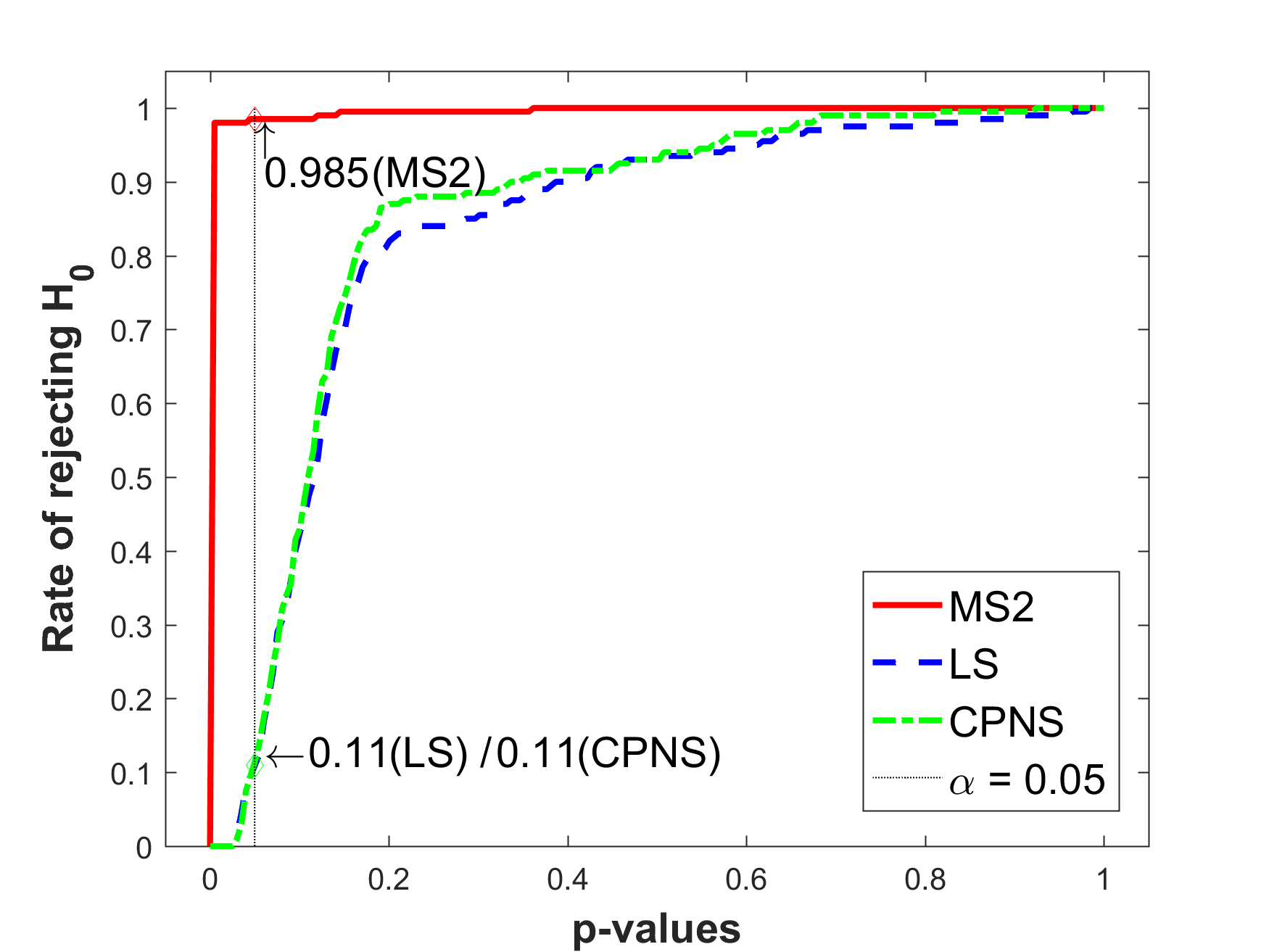}
        \caption{}
    \end{subfigure}
    \caption{(a) Two s-rep models of randomly-bent ellipsoids. Skeletal mesh points (blue) due to grid lines (purple) with spokes (green, red and cyan). (b) Directions-circles plot: Graphical display of MS2 parameter estimates (small-circles (yellow), $\hat\mu_0$ (blue dashed axis) compared with $\mu_0^*$ (black dashed)) laid over the data where different colors correspond to different observations. (c) Histogram of estimated ``horizontal'' correlation coefficients. (d) Empirical distributions of p-values from horizontal-dependence tests. See text for details.
     \label{fig:ellipsoids-all}}
	\end{center}
\end{figure}

\subsection{Data preparation}
For our purpose of validating the use of the proposed distributions, we use an s-rep data set, fitted from $30$ deformed ellipsoids; two samples from this data set are shown in Fig.~\ref{fig:ellipsoids-all}(a). This data set was previously used in \cite{Schulz2015} as a simple experimental representation of real human organs.
The data set was generated by ``physically bending'' a template ellipsoid about an axis $\mu^*_0 = (0,1,0)$ by random angles drawn from a normal distribution with standard deviation 0.4 (radians). Each deformed ellipsoid is then recorded as a 3-dimensional binary image. To mimic the procedure of fitting s-reps from, for example, medical resonance imaging of a real patient, s-reps with $74$ spokes were fitted to these binary images. (See \cite{Pizer2013} for details of the s-rep fitting.) As a preprocessing, we chose $K=58$ spoke vectors, excluding the vectors with very small total variation.

\subsection{Inference on the bending axis} 

Fitting the iMS1 distribution, we obtained axis estimate $\hat{\mu}_0^{(\mathrm{iMS1})} = (0.007,1.000,-0.008)$ (rounded to three decimal digits). Similarly, from the MS2 fitting, $\hat{\mu}_0^{(\mathrm{MS2})} = (0.006,1.000,0.006)$.
These estimates are virtually the same, only $0.6$ degrees away from the ground truth $\mu_0^*$.
Estimates of the concentric small-circles $\mathcal{C}(\hat{\mu}_0^{(\mathrm{MS2})},\hat\nu_j)$ for four choices of $j$ (the spoke index) are also shown in the top right panel of Fig.~\ref{fig:ellipsoids-all} in which $\hat{\mu}_0^{(\mathrm{MS2})}$ and $\mu_0^*$ are also shown. The test of axis (Hypothesis 2) discussed in Section~\ref{sec:testing} is applied to test $H_0:\, \mu_0 = \mu_0^*$ under the iMS2 model.
With the test statistic assuming the value $W_n=2.69$, compared to 5.99 (the 95th percentile of $\chi_2^2$), and the corresponding p-value of 0.26, we cannot reject that the s-rep spokes are rotated about the bending axis $\mu_0^*$.

\subsection{Inference on horizontal dependence}
An advantage of modeling the s-rep spoke directions by the MS2 distribution is the ability of perceiving and modeling the horizontal dependence among directions. 
As an exploratory step, we have collected the estimated correlation coefficients, computed from the approximate precision matrix $\widehat\Sigma^{-1}$, whose elements are $\hat\kappav_1$ and $\widehat\Lambdav$; see (\ref{eq:largeconcentration}). A histogram of $K(K-1)/2$ estimated correlation coefficients is plotted in the bottom left panel of Fig.~\ref{fig:ellipsoids-all}. Notably, pairs of spoke vectors from the same side (e.g. two spoke vectors in the ``left side'' of the ellipsoids in Fig.~\ref{fig:ellipsoids-all})  exhibit strong positive correlations, while those from the opposite sides exhibit strong negative correlations. The horizontal dependence is in fact apparent by the way data were generated (simultaneously bending all the spoke directions).

For large enough sample sizes, we could use the test of association discussed in Section~\ref{sec:testing} for testing $H_0:\, \bLambda = \textbf{0}$. Unfortunately, due to our small sample size, $n =30$, and the large number of parameters tested,  1653 ($= K(K-1)/2$), this is infeasible. Coping with this high-dimension, low-sample-size situation is beyond the scope of current paper, and we resort to choose only two spoke directions to test the dependence, but to repeat the testing for many different pairs of total $K = 58$ spokes. For each pair of spokes, the likelihood-ratio test produces a p-value for the pair. Investigating the empirical distribution of these p-values can provide a rough estimate of the power. In Fig.~\ref{fig:ellipsoids-all}(d), it can be seen that, at the significance level $0.05$, the MS2 test of dependence is indeed powerful, with a rejection rate of 97\%.

To provide some context to this rate, the MS2 test was compared with other natural choices of tests. We applied two methods that were previously used for s-rep data analysis: the composite principal nested spheres (CPNS), discussed in \cite{Pizer2013}, and the least-square (concentric) small-circle fitting method of \cite{Schulz2015}.

The CPNS-test is built as follows. First, the least-square small-circle is fitted to each marginal directions on $\mathbb{S}^2$. With an understanding that the axis of the fitted small-circle points to the north pole, the observations (say, $x_{i(k)}$ from the $i$th sample, $k$th spoke) are represented in spherical coordinates ($\theta_{i(k)},\phi_{i(k)}$).
For the purpose of testing ``horizontal associations'', we only keep the longitudinal coordinates $\theta_{i(k)}$. For any given pair ($k, \kappa$), Fisher's z-transformation is used to obtain the p-value in testing whether the correlation coefficient between $\theta_{i(k)}$ and $\theta_{i(\kappa)}$ is zero. We refer to this test procedure by a CPNS test.

An LS test procedure is defined similarly to the CPNS test, except that the first step of fitting individual small-circles is replaced by fitting \textit{concentric} small-circles.

These two tests were also conducted for the same combinations of spoke directions, and the empirical distributions of respective p-values are also plotted in Fig.~\ref{fig:ellipsoids-all}. These alternative tests appear to be too conservative, with rejection rates 11\% for the LS test, and 13.5\% for the CPNS test (at level 0.05).  Heuristically, the higher power of the MS2 tests is due to the superior fitting of the MS2 distribution. In particular, the ``horizontal angles'' predicted from the MS2 tend to be linearly associated, while those from the least-squares fit tend to be arbitrary. We refer to the online supplementary material for more simulation results.
All in all, using the MS2 distribution shows a clear advantage in modeling and testing the horizontal dependence of multivariate directions.

\section{Human knee gait analysis}\label{sec:applications.knee}

In biomechanical gait analysis, accurately modeling human knee motion during normal walking has a potential to differentiate diseased subjects from normal subjects. In particular, the axis of bending (of the lower leg toward the upper leg) is believed to be a key feature in the discrimination among the diseased and normal subjects \citep{Pierrynowski2010}. As a step towards the development of statistical tests for a  ``two-group axes difference," in this section we employ the proposed distributional families in modeling the bending motion of the knee.

The raw data set we use is obtained from a healthy volunteer and it is a time series of coordinates of markers planted at the volunteer's leg, recorded for 16 gait cycles.
For each time point, the directional vectors on $(\mathbb{S}^2)^5$ were computed to be the unit vectors between reference markers, as done in \cite{Schulz2015}.
These directional vectors are  \emph{horizontally dependent} of each other (as evidenced in Fig.~\ref{fig:knee.depS2}), which suggests that we may fit the MS2 distribution.

The first panel of Fig.~\ref{fig:knee.depS2} illustrates the result of MS2 fit to the all data points.
There, we superimpose the fitted concentric circles to the observed directional vectors, including their estimated axis, together with a hypothesized dominant bending axis $\mu_0^* = (0,1,0)^\top$, the left-right axis of the subject. The MS2 model seems to fit well with high estimated horizontal correlation coefficients. We, however, identify a strong evidence against using a single MS2. Specifically, as shown in Fig.~\ref{fig:knee.depS2} some directional vectors exhibit higher variations for a subset of time points.

In fact, the data consist of many inhomogeneous periods of the gait cycle. We focus on the ``swing'' and ``stance'' periods, and separately analyze subsampled data from each period. The MS2 model fits well for the swing period data (see Fig.~\ref{fig:knee.depS2}(b)), and the estimated axis $\hat{\mu}_{0}^{(\mbox{sw})}= (0.013, 1.000, 0.005)^\top$ is only 0.8 degrees away from the hypothesized axis, $\mu_{0}^*$. As expected, our likelihood ratio test procedure does not reject the null hypothesis $H_0: \mu_0 = \mu_0^*$, with p-value 0.16.
For the stance period data, excluding the highly-irregular directions shown as dark blue points in Fig.~\ref{fig:knee.depS2}(c), the MS2 model also fits well, and we confirm that the axis of bending for this period differs from $\mu_{0}^*$, with p-value less than $10^{-5}$. The estimated axis for the stance period is $\hat{\mu}_{0}^{(\mbox{st})} = (0.11, 0.994, 0.006)^\top$.


While the MS2 distribution was useful in making inference on the bending axis of partial knee motions, future work for this type of data lies in the development of a two-sample axis difference test.  
  
\begin{figure}[!t]
	\begin{center}
	\begin{subfigure}{0.3\textwidth}
		\includegraphics[width = \textwidth]{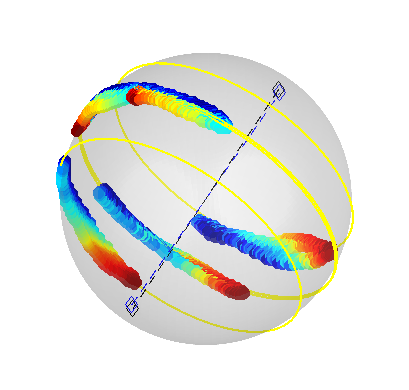}
		\caption{All data points}
	\end{subfigure}
	\begin{subfigure}{0.3\textwidth}
		\includegraphics[width = \textwidth]{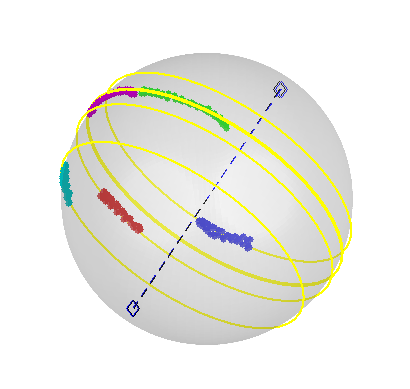}
		\caption{Swing period}
	\end{subfigure}
	\begin{subfigure}{0.3\textwidth}
		\includegraphics[width = \textwidth]{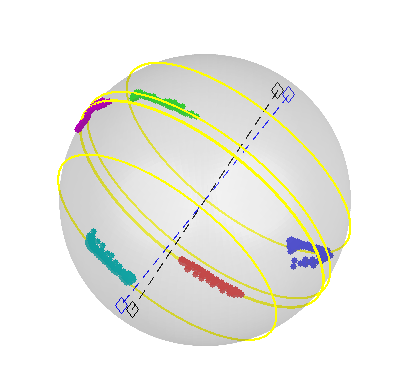}
		\caption{Stance period}
	\end{subfigure}
    \caption{Knee gait data:  Observed directional vectors overlaid with the hypothesized (black dashed)  and estimated  (blue dashed) axes as well as the MS2-fitted small circles. The directions along the north-most circle in the stance period exhibit a higher and irregular pattern of variation. In (a), colors code time indices. In (b) and (c), different colors represent different $\mathbb{S}^2$. The data are horizontally associated, with correlation coefficients ranging from 0.17 to 0.97 (in absolute values), which are significantly different from zero. \label{fig:knee.depS2}}
	\end{center}
\end{figure}

\section*{Acknowledgements}
Stephan Huckemann gratefully acknowledges funding by the Niedersachsen Vorab of the Volkswagen Foundation and DFG HU 1575/4.

\section*{Supporting information} Additional information for this article is available online. These include graphics of S2 densities (Appendix A1), powers of the test (Appendix A2), and additional simulation results (Appendix A3 and A4).

\bibliographystyle{asa}
\bibliography{library}

\section*{Address}

Department of Statistics, University of Pittsburgh, 1806 Wesley W. Posvar Hall, 230 Bouquet Street, Pittsburgh, PA 15260.
 \texttt{sungkyu@pitt.edu}

\appendix
\section*{Appendix}

We provide a technical lemma, referenced in Section \ref{sec:estimation.S1}, and proofs of Propositions \ref{prop:invariance} and \ref{prop:S1constant}.

\begin{lem} \label{lem:ambientMeanLemma}
If $X \sim $S1($\mu_0,\mu_1,\kappa_0,\kappa_1$), then $E(X)$ is a linear combination of $\mu_0$ and $\mu_1$.
\end{lem}

\begin{proof}[Proof of Lemma~\ref{lem:ambientMeanLemma}]
Suppose that for some $a,b,c \in \Real$, $v \in S^{p-1}$, $E(X) = a \mu_0 + b \mu_1 + c v$.
Then choose a $B \in O(p)$ such that $B\mu_0 = \mu_0$, $B\mu_1 = \mu_1$ but $B v \neq v$.
By Proposition~\ref{prop:invariance}(i), $BX \sim S1(\mu_0,\mu_1, \kappa_0,\kappa_1)$.
Thus $E(X) = E(BX)$, which in turn leads to
   $a \mu_0 + b \mu_1 + c v = a \mu_0 + b \mu_1 + c Bv$,
which is true only if $c = 0$. This gives the result.
\end{proof}

\begin{proof}[{Proof of Proposition~\ref{prop:invariance}}]
Assertions (i) and (iii) are trivial. (ii): Assume that $X\sim BX$. With the representation (\ref{eq:fbdensity}) this is equivalent with
$$ (B \gamma)^\top x = \gamma^\top x\mbox{ and } \big((B \mu_0)^\top x\big)^2 = (\mu_0^\top x)^2\mbox{ for all }x\in \mathbb{S}^{p-1}\,.$$
In consequence
$$2\nu\kappa_0 B\mu_0 + \kappa_1B\mu_1 = B\gamma = \gamma =2\nu\kappa_0\mu_0 + \kappa_1\mu_1\mbox{ and }B\mu_0 = \pm \mu_0\,.$$
Thus the case of $B\mu_0 =\mu_0$ yields at once $B\mu_1 = \mu_1$. In contrast, the case  $B\mu_0 =-\mu_0$ yields that
$ \|B\mu_1\|^2 = \|4\nu \frac{\kappa_0}{\kappa_1}\mu_0 +\mu_1\|^2 = 1 + 8\nu^2\kappa_0/\kappa_1 (1 + 2 \kappa_0/\kappa_1) >1$, which cannot be if $B$ is orthogonal.

Further, if $Y\sim BY$, since the exponent is again the sum of an even function in $x$ and an odd function, we have again that $B\mu_0 = \pm \mu_0$, yielding  $\|P_{\mu_0}x\| = \|P_{\mu_0}Bx\|$, due to orthogonality, $B$ preserves the space orthognal to $\mu_0$,  as well as
$$ 2\nu \kappa_0 \mu_0^\top x + \frac{\kappa_1}{\|P_{\mu_0}\mu_1\|\,\|P_{\mu_0}x\|} \mu_1^\top (I_p-\mu_0\mu_0^\top)x =  2\nu \kappa_0 (B\mu_0)^\top x + \frac{\kappa_1}{\|P_{\mu_0}\mu_1\|\,\|P_{\mu_0}x\|} \mu_1^\top (I_p-\mu_0\mu_0^\top)B^\top x\,.$$
As before, $B\mu_0 =\mu_0$ yields that $B\mu_1 =\mu_1$ and $B\mu_0 =-\mu_0$ would give that
$$ \big((B\mu_1)^\top - \mu_1^\top) x = 2\nu \left(\frac{2\kappa_0}{\kappa_1} \sqrt{1-\nu^2}\sqrt{1-(\mu^\top_0x)^2} - 1\right)\mu_0^\top x\mbox{ for all }x\in \mathbb{S}^{p-1}\,.$$
In particular, this would yield that $B\mu_1-\mu_1$ is a multiple of $\mu_0$, the factor, however, is not constant in $x$, a contradiction.
\end{proof}

\begin{proof}[Proof of Proposition \ref{prop:S1constant}]
For a given $h > 0$, let $\gamma = 2\kappa_0\nu\mu_0 + \kappa_1\mu_1$ and $A_h = \kappa_0\mu_0\mu_0^\top + hI_p$.  Then it is easy to see that
the S1 density  (\ref{eq:S1density}) can be expressed as the Fisher-Bingham form (\ref{eq:fbdensity}):
\begin{eqnarray*}
f_{\textrm{S1}}(x; \mu_0, \mu_1, \kappa_0, \kappa_1)= \alpha(\gamma, A_h) \exp\{\gamma^\top x  -x^\top A_h x  \}
\end{eqnarray*}
where $\alpha(\gamma, A_h)$ satisfies
\begin{equation}\label{eq:constant.S1.FB}
	a(\kappa_0, \kappa_1, \nu) = \alpha(\gamma, A_h) \exp\{ -\kappa_0\nu^2 + h\}.
\end{equation}

For the purpose of evaluating the value of $a(\kappa_0, \kappa_1, \nu)$, or equivalently $\alpha(\gamma, A_h)$ for the given value of $h$, one can assume without losing generality that $\mu_0 = (1,0,\ldots,0)^\top$ and $\mu_1 = (\nu, \sqrt{1-\nu^2},0,\ldots,0)$, so that $\gamma = (\nu(2\kappa_0 + \kappa_1), \kappa_1\sqrt{1-\nu^2}, 0, \ldots, 0)^\top$, and the vector of diagonal values of $A_h$ are $\lambda := (2(\kappa_0 + h), h,\ldots, h)$.
The $j$th element of $\xi$, in the statement of proposition, is then given by $\mu_j := \gamma_j / 2\lambda_j$. With these notations, Proposition 1 of \cite{Kume2005} gives
\[
	\alpha(\Psi, \mu) = 2\pi^{p/2} |A_h|^{-1/2} g(1) \exp\{ \xi^\top A_h\xi \}.
\]
Hence, by (\ref{eq:constant.S1.FB}), we have (\ref{eq:norm.const.FB2}).
\end{proof}

\newpage 


\section*{Supplementary material}

*Supporting information for*

# Small sphere distributions for directional data with application to medical imaging

Byungwon Kim[1], Stephan Huckemann[2], Jörn Schulz[3] and Sungkyu Jung[1]

[1]Department of Statistics, University of Pittsburgh

[2]Felix-Bernstein-Institute for Mathematical Statistics in the Biosciences, University of Göttingen

[3]Department of Electrical and Computer Engineering, University of Stavanger

May 27, 2017

## A1   Graphics of S2 densities (referenced in Section 2.2)

Figure A1 illustrates the small-sphere densities of the second kind (S2), for the same set of parameters used in Figure 2 of the main article.

The S2 density is relatively high near the small circle $\mathcal{C}(\mu_0, \nu)$ and has the mode at $\mu_1 \in \mathcal{C}(\mu, \nu)$. The S1 and S2 densities with the same parameters look similar to each other.

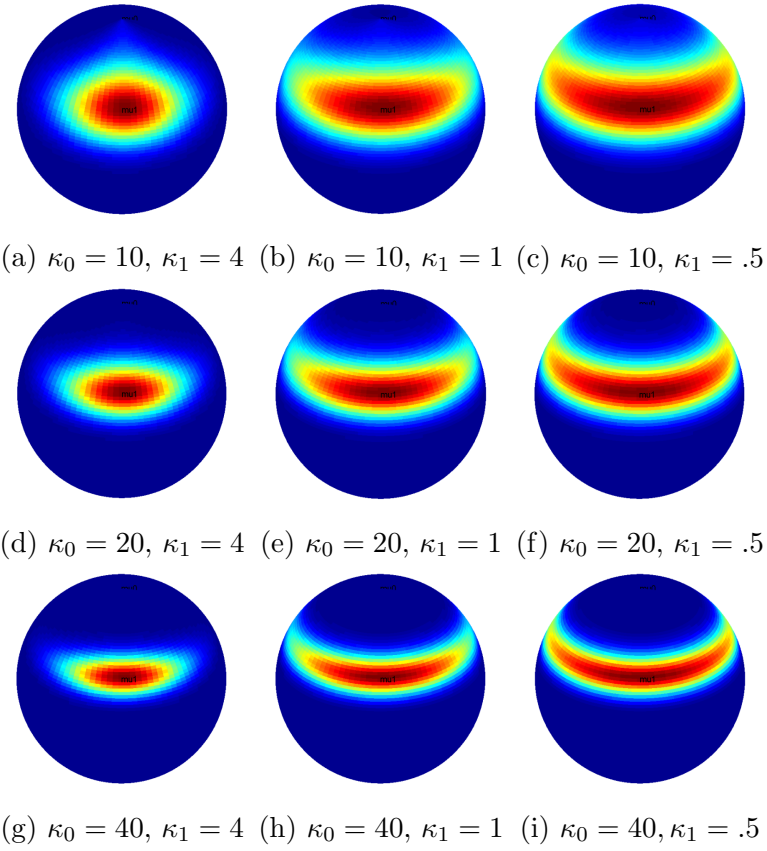

Figure A1: Illustration of S2 densities with several choices of $\kappa_0$ and $\kappa_1$.

# A2 Null distributions and empirical powers of tests (referenced in Section 4)

**Null distribution of $W_n$.** As referenced in Section 4 of the main article, we demonstrate that the empirical null distributions of the test statistics $W_n$ are indeed approximately chi-square distributed where we use a moderate sample size $n = 30$. To check this, we use Q-Q envelope plots (Lee, 2007) of test statistics under the null hypotheses discussed in Section 4 of the main article, which, for convenience, we restate. Note that in Hypothesis 1 we consider only $p = 3$.

1. **Test of association**. $H_0$: $\boldsymbol{\Lambda} = \mathbf{0}$, or $\theta \in \Theta_0 = \mathbb{S}^{p-1} \times (\mathbb{S}^{p-1})^K \times (\Re_+)^K \times (\Re_+)^K \times \{\mathbf{0}\}$. Under $H_0$, the model degenerates to the iMS2 and there is no horizontal dependence.

2. **Test of axis**. $H_0$: $\mu_0 = \mu_0^*$, or $\theta \in \Theta_0 = \{\mu_0^*\} \times (\mathbb{S}^{p-1})^K \times (\Re_+)^K \times (\Re_+)^K \times (\Re)^{K(K-1)/2}$. This is to test whether a predetermined axis $\mu_0^*$ of the small sphere is acceptable.

3. **Test of great-sphere**. $H_0$: $\boldsymbol{\nu} = 0$, or $\theta \in \Theta_0 \simeq \mathbb{S}^{p-1} \times (\mathbb{S}^{p-2})^K \times (\Re_+)^K \times (\Re_+)^K \times (\Re)^{K(K-1)/2}$. ($A \simeq B$ means that $A$ and $B$ are diffeomorphic.) This is to test whether the underlying spheres are great spheres with radius 1.

4. **Test for von Mises-Fisher distribution**. $H_0$: $\kappa_0 = 0$, or $\theta \in \Theta_0 \simeq \mathbb{S}^{p-1} \times \Re_+$. Under $H_0$, there is no "small-circle feature."

5. **Test for Bingham-Mardia distribution**. $H_0$: $\kappa_1 = 0$, or $\theta \in \Theta_0 \simeq \mathbb{S}^{p-1} \times \Re_+$. Under $H_0$, there is no unique mode.

We provide empirical null distributions for Hypotheses 2–5 (labeled as (a)–(d)) while the full model is the S1 distribution. The size of the test of association (Hypothesis 1) can be checked in Fig. A3. The Q-Q envelope plots for each of the test statistics are shown in Fig. A2.

In each panel of Fig. A2, the Quantile-Quantile plot of $W_n$ (with respect to the asymptotic null distribution), simulated under each corresponding null hypothesis, is shown as the red curve. This is overlaid with 100 Q-Q plots (shown in blue curves), obtained from the random samples of the same size, following the theoretical chi-squared distribution. The blue curves provide an envelope, representing the natural variation of $\chi^2$ samples. Based on Fig. A2, arguing as in Lee (2007), we conclude that the test statistic $W_n$ approximately follows the chi-square distribution.

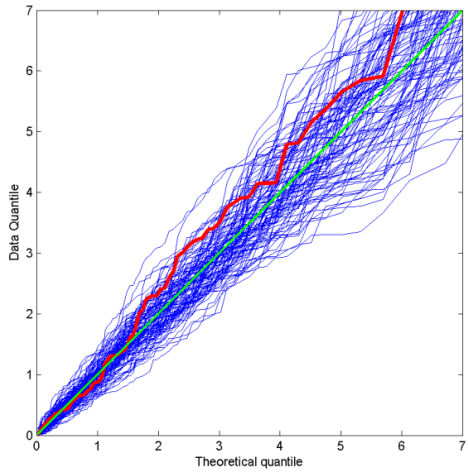
(a) $H_0$: $\mu_0 = \mu_0^*$

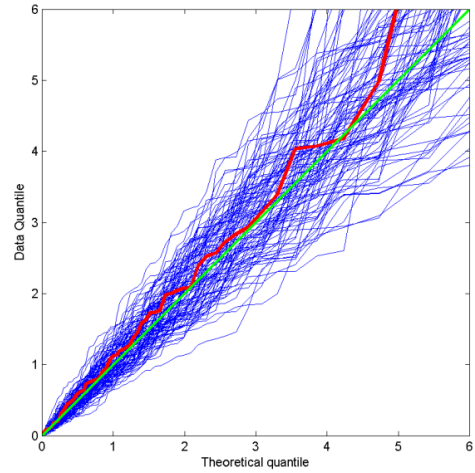
(b) $H_0$: $\nu = 0$

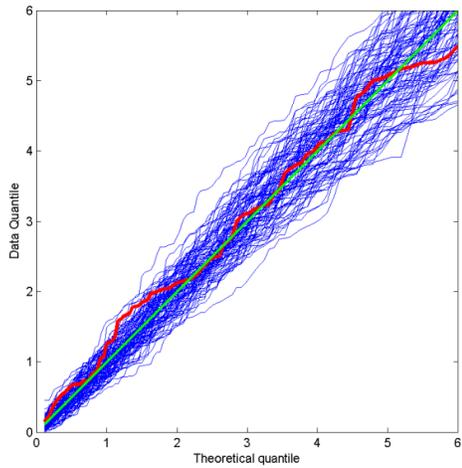
(c) $H_0$: $\kappa_0 = 0$

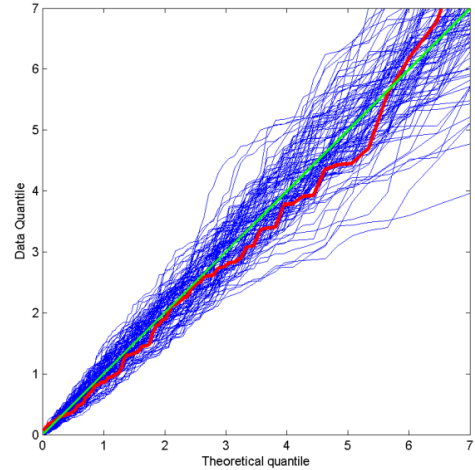
(d) $H_0$: $\kappa_1 = 0$

Figure A2: Q-Q envelope plots for testing the asymptotic distribution of $W_n$ against $\chi^2_{\text{df}}$, where df $= q_1 - q_2$ denotes the corresponding degrees of freedom. The sample size is $n = 30$. For each case, the Q-Q plot of the test statistics (the red curve) is inside the acceptable variation (given by the envelope within the blue lines). See text and Lee (2007) for the use of the Q-Q envelope plot.

**Empirical powers of the tests.** We provide empirical powers of the proposed likelihood ratio tests for the test of associations in MS2 (Hypothesis 1) and Hypotheses 3 and 5 under S1. The powers of the test for von Mises-Fisher distribution (Hypothesis 4) are also reported in Section 5.3 of the main article. In all simulations, we used 200 repetitions to compute the empirical rejection rates, at significance level 0.05.

- **Test of association (MS2).** Hypothesis 1. See Fig. A3. This is to test $H_0$: iMS2 vs. $H_1$: MS2. In other words, $H_0$: $\lambda_{jk} = 0$ for all $j \neq k$, $j, k = 1, \ldots, K$. Empirical powers for $K = 2$ and 3 are shown in Fig. A3, for various alternative settings. We reparametrize $\Lambda$ using the correlation coefficient $\rho$, as follows. For $K = 2$, the parameters $(\kappa_0, \kappa_1, \lambda_{12})$ are parameterized by $\sigma_v, \sigma_h, \rho$, representing the vertical standard deviation, horizontal standard deviation, and horizontal correlation coefficient;

$$\kappa_0 = (2\sigma_v^2)^{-1}, \quad \begin{pmatrix} \kappa_{1(1)} & -\lambda_{12} \\ -\lambda_{12} & \kappa_{1(2)} \end{pmatrix} = \begin{pmatrix} \sigma_h^2 & \rho\sigma_h^2 \\ \rho\sigma_h^2 & \sigma_h^2 \end{pmatrix}^{-1}. \quad (1)$$

Figure A3 shows that for the cases $K = 2$ and 3 the power sharply grows from zero to 1, with a notably hike at $\rho = 0.3$. The empirical type I error rates are controlled below the significance level 0.05. The figure is generated for sample size $n = 100$ from the MS2 distribution with the fixed vertical ($\sigma_v^2 = 0.005$) and horizontal ($\sigma_h^2 = 1$) dispersions but with several values of the association parameter ($\rho = 0, 0.1, \ldots, 0.9$).

- **Test of great-sphere (S1).** Hypothesis 3. See Fig. A4. The true parameters are set so that the angles in the small circle are either 90° (i.e., $H_0$ is true), 80°, 70° and 60° (the latter three fall under the alternative hypothesis $H_1$). These situations are denoted as cases (a)-(d) in Fig. A4, where we visualized a sample of size 100 for each of the cases, in order to give a visual impression of the "effect size".

- **Test for Bingham-Mardia distribution (S1).** Hypothesis 5. See Fig. A5. In the figure, the null distribution (the BM distribution) is shown in (a). Three different alternative distributions are shown in (b)–(d).

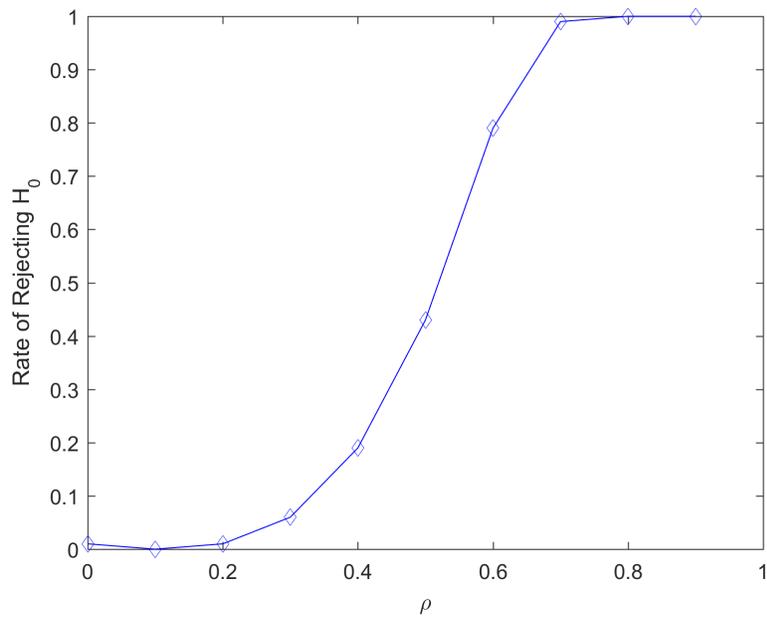

(a) $K = 2$ (MS2)

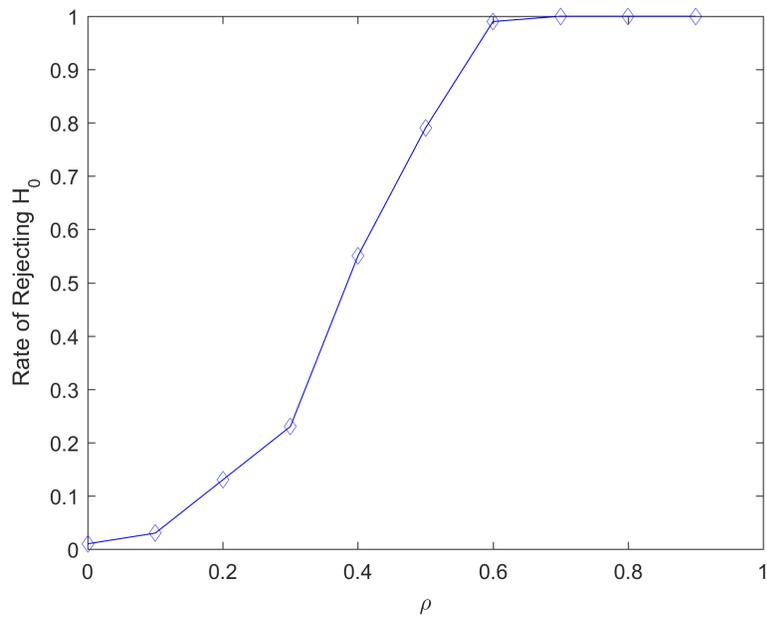

(b) $K = 3$ (MS2)

Figure A3: Empirical powers of the test of association at significance level 0.05. The rates of rejecting $H_0 : \lambda_{jk} = 0$ for all $j \neq k$, $j, k = 1, \ldots, K$ computed from 200 repetitions are shown.

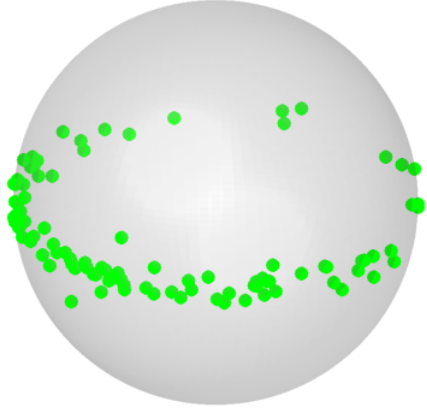 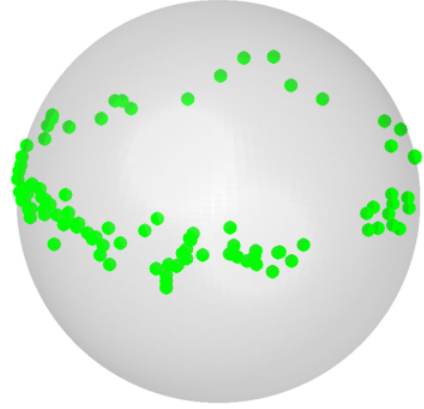

(a) S2 (arccos($\nu$) = 90°)  (b) S2 (arccos($\nu$) = 80°)

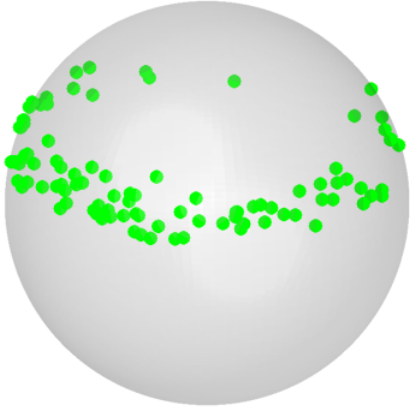 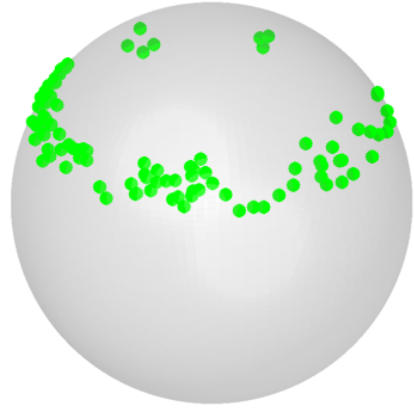

(c) S2 (arccos($\nu$) = 70°)  (d) S2 (arccos($\nu$) = 60°)

Figure A4: Empirical power $\hat{\beta}$ of testing for a great-sphere (S1): $\hat{\beta}$ = 0.045, 1, 1, 1 for cases (a), (b), (c) and (d), respectively. Examples of random samples of size $n = 100$ are shown.

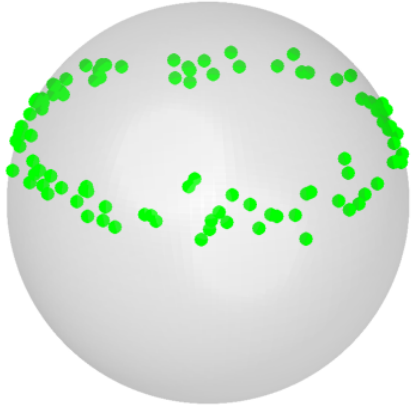
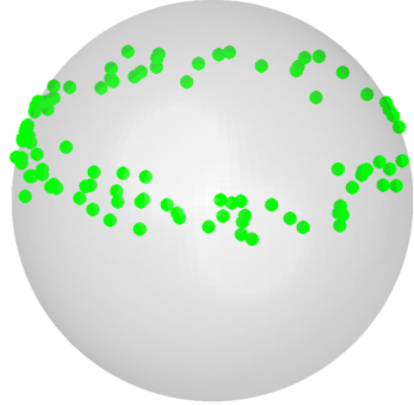

(a) BM ($\kappa = 100$)
S2 ($\kappa_0 = 100, \kappa_1 = 0$)

(b) S2 ($\kappa_0 = 100, \kappa_1 = 0.5$)

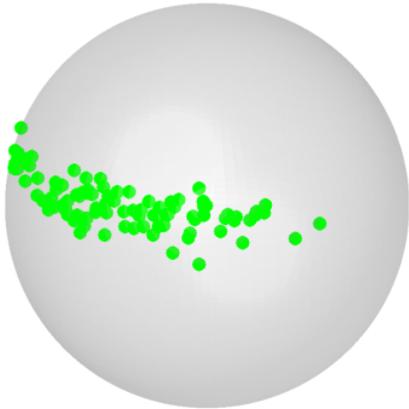
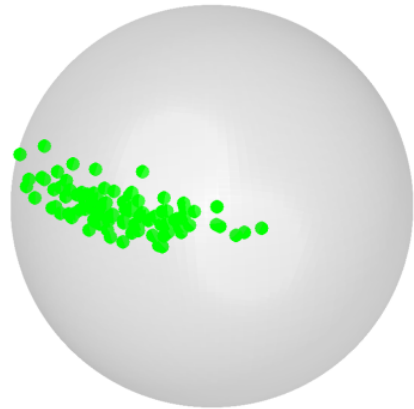

(c) S2 ($\kappa_0 = 100, \kappa_1 = 5$)

(d) S2 ($\kappa_0 = 100, \kappa_1 = 10$)

Figure A5: Empirical power $\hat{\beta}$ of testing for the Bingham-Mardia distribution (S1): $\hat{\beta} = 0.065, 0.925, 1, 1$ for cases (a), (b), (c) and (d), respectively. Examples of random samples of size $n = 100$ are shown.

8

# A3 Additional simulation results (referenced in Section 5.1)

As referenced in Section 5.1 of the main article, we provide a supplementary figure and additional simulation results.

**Bivariate data on $(\mathbb{S}^2)^2$.** Figure A6 illustrates random samples from settings (a)–(f) used in Tables 3 and 4 in the main article.

**Robustness against model misspecification.** To check robustness against model misspecification, we use a *signal-plus-noise* model that is neither S1 nor S2. This model generates observations on $(\mathbb{S}^2)^K$, that are rotated from a reference point, and perturbed by a spherical noise. Given $\mu_0, \mu_{1j} \in \mathbb{S}^2$, the perturbation model for an observation $x_j \in \mathbb{S}^2$ is

$$x_j = R(\mu_0, \theta_j)\mu_{1j} \oplus \epsilon_j \in \mathbb{S}^2, \quad j = 1, \ldots, K, \qquad (2)$$

where $R(\mu_0, \theta_j)$ is the rotation matrix so that $R(\mu_0, \theta_j)\mu_{1j}$ gives the rotation of the vector $\mu_{1j}$ by an angle $\theta_j$ about the axis $\mu_0$. The action $\oplus$ is defined as $v \oplus \epsilon = (v + \epsilon)/\|v + \epsilon\|$. The spherical error $\epsilon_j$ is independent of $\boldsymbol{\theta}$ and is sampled from $\mathrm{N}(\mathbf{0}, \sigma_v^2 I_3)$ for $\sigma_v^2 > 0$. In the univariate case, $\theta_1$ is sampled from a normal distribution with standard deviation $\sigma_h$, and for the bivariate case, $(\theta_1, \theta_2)$ is sampled from a bivariate normal with the precision matrix specified in (1).

Just like Tables 2 and 3 in the main article, we report means and standard deviations of the *angular product errors* in Table A1 (univariate cases) and Table A2 (multivariate cases). In Tables A1 and A2, the dispersion parameters $(\sigma_v^2, \sigma_h^2, \rho)$ are carefully chosen so that Case (a) (or b, c, d) of Table A1 corresponds to Case (a) (or b, c, d, respectively) of Table 2 in the main article. Likewise Case (a) (or b, c, d) of Table A2 corresponds to Case (a) (or b, c, d, respectively) of Table 3 in the main article.

Our estimates perform rather well against model misspecification. The performance of our estimates is comparable to that of the least-squares estimates, which is designed to estimate the parameters of the signal-plus-noise model. The angular product errors from our estimates (S1, S2, iMS1, iMS2 and MS2) in Tables A1 and A2 are comparable to the least-squares estimates, and are often smaller than the BM estimates.

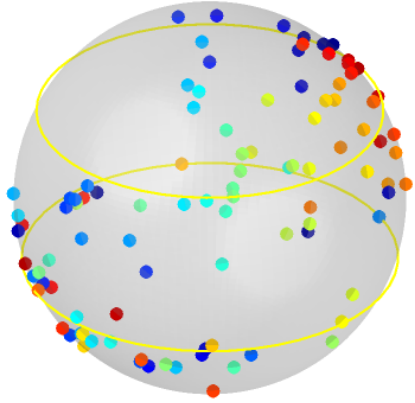
(a) $(\kappa_{0j}, \kappa_{1j}, \lambda_{12}) = (10, 1, 0)$

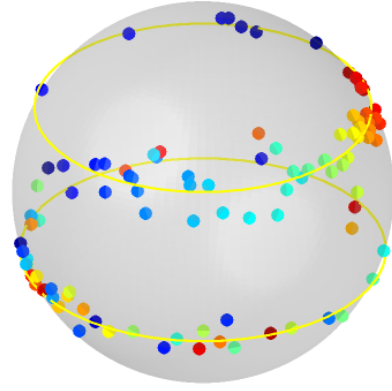
(b) $(\kappa_{0j}, \kappa_{1j}, \lambda_{12}) = (100, 1, 0)$

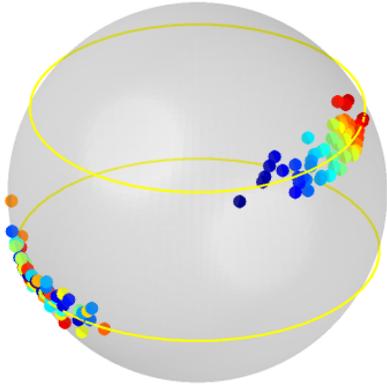
(c) $(\kappa_{0j}, \kappa_{1j}, \lambda_{12}) = (100, 10, 0)$

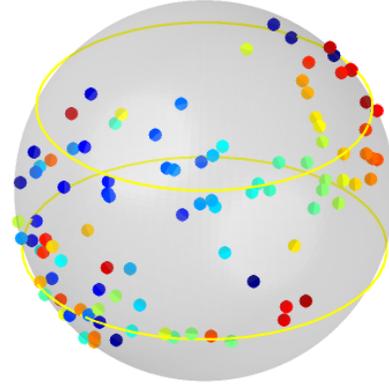
(d) $(\kappa_{0j}, \kappa_{1j}, \lambda_{12}) = (10, 2, 1.5)$

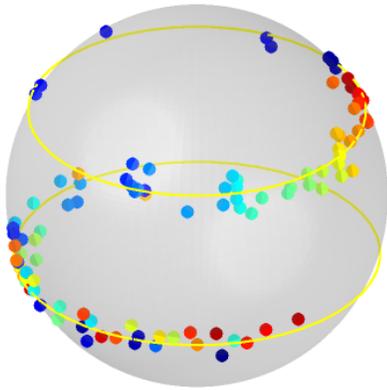
(e) $(\kappa_{0j}, \kappa_{1j}, \lambda_{12}) = (100, 2, 1.5)$

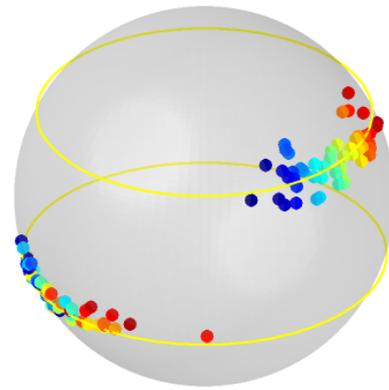
(f) $(\kappa_{0j}, \kappa_{1j}, \lambda_{12}) = (100, 20, 15)$

Figure A6: Random samples of size $n = 50$ from the MS2 model on $(\mathbb{S}^2)^2$ used in our simulations. Different colors represent different observations. Models (a)–(c) are from iMS2, while models (d)–(f) are from the MS2, with $\lambda_{12} > 0$ (positive association parameters).

| Method | (a) | (b) | (c) |
|---|---|---|---|
| S1 | 7.06(3.74) | 1.67(1.16) | 10.70(10.77) |
| S2 | 7.05(3.89) | 1.67(1.08) | **10.32(9.31)** |
| BM | 10.09(7.01) | 1.80(1.25) | 11.91(10.78) |
| LS | **6.46(3.50)** | **1.66(1.10)** | 10.93(8.96) |

Table A1: Small-circle estimation performances for univariate data on $\mathbb{S}^2$. Means (standard deviations) of the angular product errors (in degrees). Data are generated from the signal-plus-noise model (2) with $K = 1$.

| Method | Independent($\rho = 0$) | | | Dependent($\rho = 0.7$) | | |
|---|---|---|---|---|---|---|
| | (a) | (b) | (c) | (a) | (b) | (c) |
| iMS1 | 5.22(1.93) | 1.51(0.67) | **3.16(3.08)** | 5.24(2.16) | 1.37(0.59) | **3.49(4.45)** |
| iMS2 | 5.15(2.43) | 1.50(0.67) | 3.78(2.24) | 5.03(2.26) | 1.36(0.58) | 3.62(2.44) |
| MS2 | 5.14(2.41) | 1.50(0.67) | 3.77(2.22) | 5.13(2.33) | 1.36(0.58) | 3.62(2.42) |
| BM | 5.54(2.03) | 1.52(0.67) | 6.84(3.70) | 6.03(3.07) | 1.37(0.60) | 6.63(3.76) |
| LS | **4.71(2.05)** | **1.45(0.68)** | 3.79(2.26) | **4.73(2.23)** | **1.32(0.58)** | 3.65(2.43) |

Table A2: Small-circle estimation performances for bivariate data on $(\mathbb{S}^2)^2$. Means (standard deviations) of the angular product errors (in degrees). Data are generated from the signal-plus-noise model (2) with $K = 2$.

# A4 Associations among s-rep spokes (referenced in Section 6.4)

As referenced in Section 6.4 of the main article, we provide more simulation results regarding the test of association, showing the advantage of using the MS2 model in analyzing s-rep data.

**Test of association applied to bi- and tri-variate directions.** In the main article, we repeatedly applied the test of association, and its competitors, for a pair of spoke directions (i.e. bivariate directions, $K = 2$). We report an additional figure (extending Fig. 6(d) in the main article) for the bivariate case, and results of power comparison for a tri-variate case. For the tri-variate case ($K = 3$), three spokes are randomly chosen among the 58 spokes of the s-rep, and the test of association is applied to compute the p-value for testing $H_0 : \mathbf{\Lambda} = \mathbf{0}$. The two competitors, LS and CPNS tests, discussed in the main article, are also applied to obtain corresponding p-values. This is repeated 200 times to obtain an empirical distribution of p-values from the MS2, LS, and CPNS tests of association. Figure A7 summarizes the results. The top row corresponds to the bivariate case (subfigure (b) is also shown in Fig. 6 of the main article); the bottom row corresponds to the tri-variate case. The left column shows histograms from 200 observed p-values, from the MS2, LS, and CPNS tests of association, respectively. The right column shows the empirical distribution functions (e.d.f.) of the p-values. If the significance level is set at $\alpha \in (0,1)$, the value of the e.d.f. at $\alpha$ gives the empirical rate of rejection at the significance level $\alpha$. For example, if $\alpha = 0.05$ in the bivariate case, shown in subfigure (b), the power of the MS2 test of association is estimated at 97%. Notably, the proposed test are much more powerful than the LS test and CPNS test (11 % and 13.5%, respectively). For $K = 3$ case, the power of the MS2 test, 100%, is higher than those of the LS test and CPNS test (10.5 % and 10%, respectively).

**Data examples for which MS2 test of association is superior.** What makes the MS2 test of association much more powerful than the other two? In the main article (Section 6.4), we write *"the higher power of the MS2 test is due to the superior fitting of the MS2 distribution."* To support this claim, we show data examples where the null hypothesis of no association is rejected by the MS2 test, but is not rejected by other tests in Fig. A8 (for the bivariate case) and Fig. A9 (for the tri-variate case). In each of these figures, we show the original data in the left column with different rows corresponding to different small-circle fittings by the MS2 model (top row), the CPNS model (middle row) and the LS model

(bottom row). Note that the MS2 and LS models fit concentric circles, while the CPNS model fits circles with generally different axis. In the middle column, the (orthogonal) projections of the data to the corresponding small-circles are shown. In essence, the MS2 test of association and the LS and CPNS tests of correlation are applied to the bi- or tri-variate horizontal angles along the small circles. To give a visual impression of the "linear" association between the horizontal angles, the scatterplots of the horizontal angles are shown in the right column. In both figures, we check that the small-circle fitting by the MS2 model (top row) provides linearly associated angles. Thus the test of association rejects the null hypothesis of no association. On the other hand, the small-circle fitting by either CPNS (middle row) and LS (bottom row) provides non-linearly associated angles (which is expected, as the both fitting procedures do not make use of the dependence structure). The test of correlation applied to these data rarely rejects the null hypothesis, resulting in the low powers of the LS and CPNS tests of association.

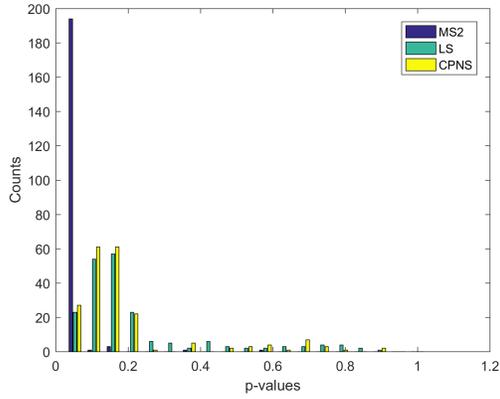
(a) $K = 2$

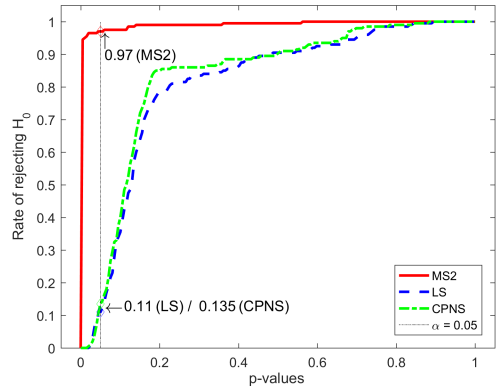
(b) $K = 2$

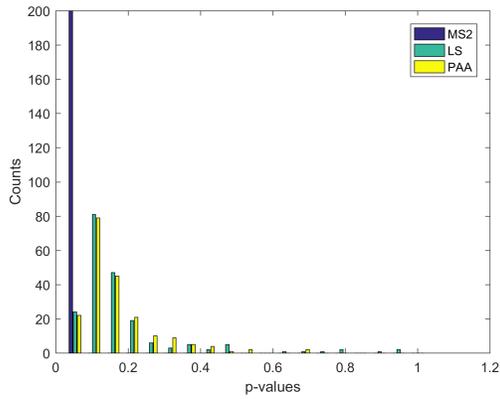
(c) $K = 3$

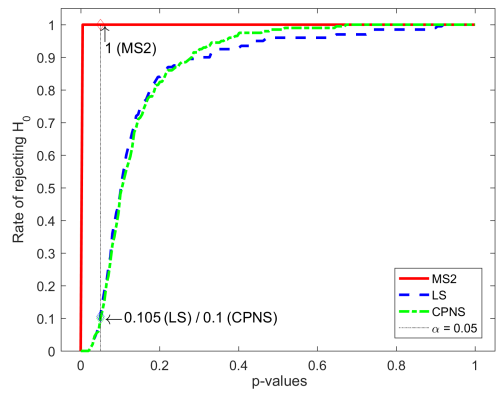
(d) $K = 3$

Figure A7: Histograms ((a) and (c)) and e.d.f.s ((b) and (d)) of p-values. The MS2 test exhibits higher powers than LS and CPNS tests.

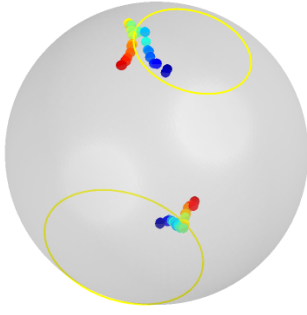
(a) MS2 - Original data and estimates

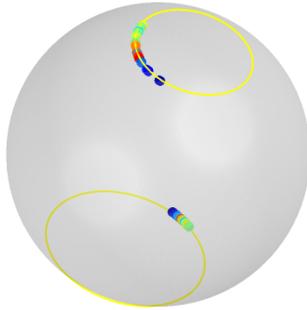
(b) MS2 - Projected data

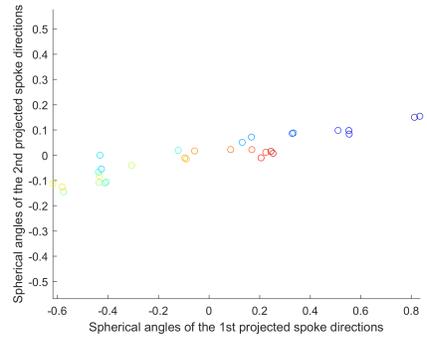
(c) MS2 - Scatter plot of horizontal angles

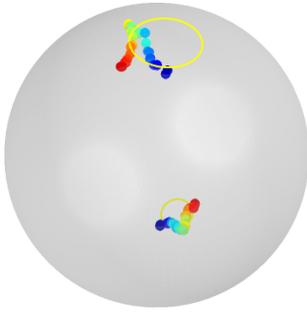
(d) CPNS - Original data and estimates

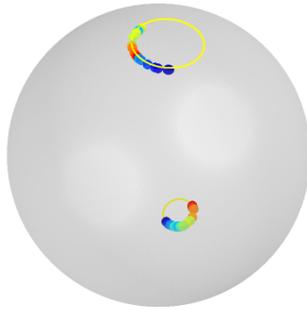
(e) CPNS - Projected data

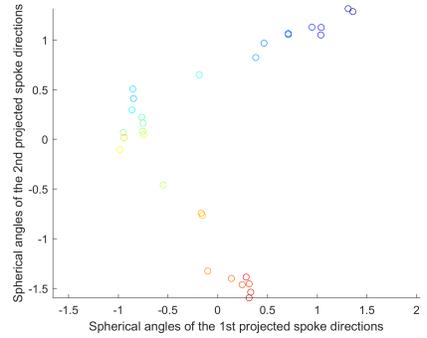
(f) CPNS - Scatter plot of horizontal angles

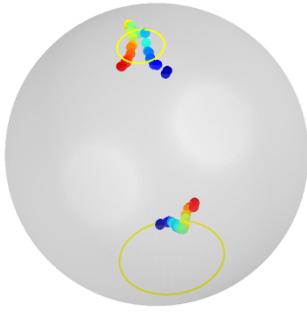
(g) LS - Original data and estimates

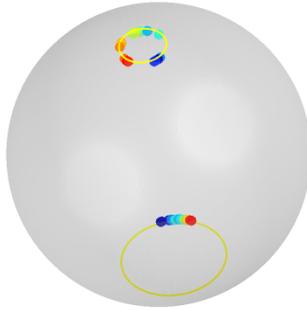
(h) LS - Projected data

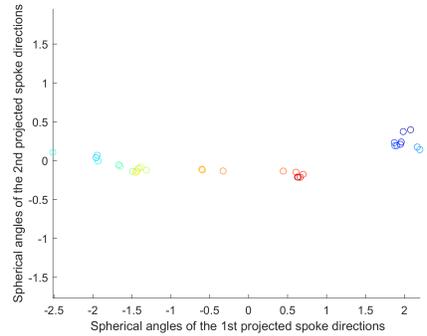
(i) LS - Scatter plot of horizontal angles

Figure A8: Data example of MS2, CPNS, and LS fittings for a bivariate data set, for which the MS2 test rejects the null hypothesis of no association, while the LS and CPNS tests do not reject. Bright yellow arcs are on the front side of the sphere, and darker yellow arcs are on the back side of the sphere. The horizontal angles predicted from the MS2 tend to be linearly associated, while those from the LS and CPNS fit tend to be arbitrary.

15

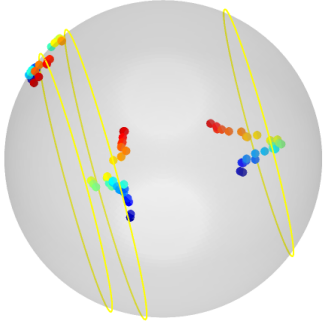
(a) MS2 - Original data and estimates

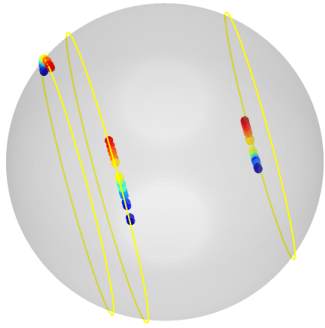
(b) MS2 - Projected data

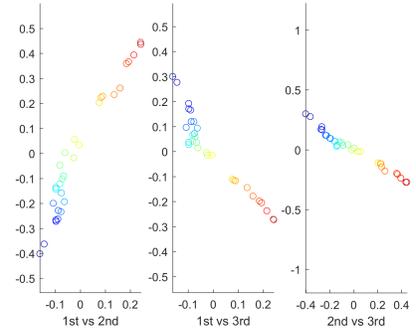
(c) MS2 - Scatter plot of horizontal angles

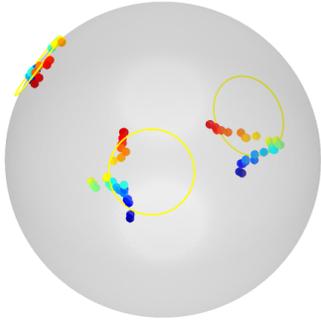
(d) CPNS - Original data and estimates

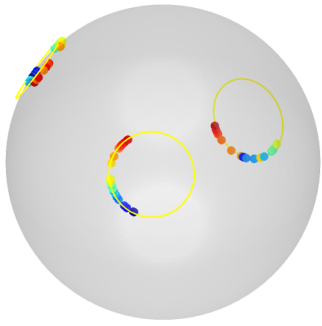
(e) CPNS - Projected data

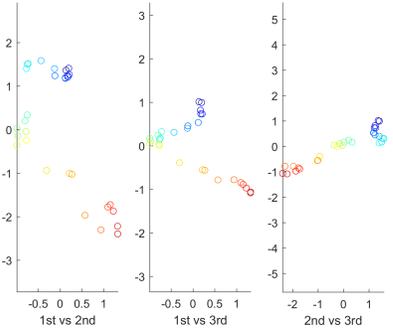
(f) CPNS - Scatter plot of horizontal angles

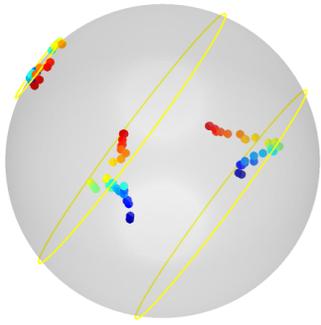
(g) LS - Original data and estimates

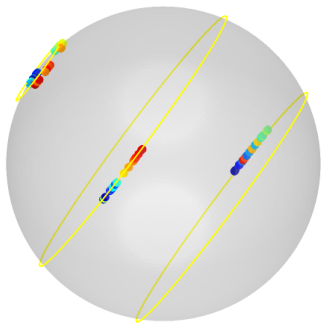
(h) LS - Projected data

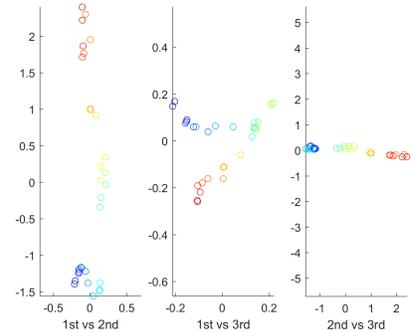
(i) LS - Scatter plot of horizontal angles

Figure A9: Data example of MS2, CPNS, and LS fittings for a trivariate data set, for which the MS2 test rejects the null hypothesis of no association, while the LS and CPNS tests do not reject. Bright yellow arcs are on the front side of the sphere, and darker yellow arcs are on the back side of the sphere. The horizontal angles predicted from the MS2 tend to be linearly associated, while those from the LS and CPNS fit tend to be arbitrary.

16



\end{document}